\documentclass[11pt]{article}

\usepackage{amsmath}
\usepackage{amsthm}
\usepackage{amssymb}
\usepackage{amsfonts}
\usepackage{subfigure}
\usepackage{color}
\usepackage{makeidx}

\newcommand{\labell}[1]{\label{#1}}

\usepackage{dsfont}
\usepackage{mathrsfs}
\newcommand{\cA}{{\cal A}}
\newcommand{\cB}{{\cal B}}
\newcommand{\cC}{{\cal C}}

\newcommand{\cM}{{\cal M}}
\newcommand{\cN}{{\cal N}}
\newcommand{\cP}{{\cal P}}
\newcommand{\cF}{{\cal F}}

\newcommand{\cR}{{\cal R}}

\newcommand{\HP}{{\cal HP}}

\renewcommand{\P}{{\cal P}}

\newcommand{\bepsilon}{\overline{\epsilon}}
\newcommand{\HLF}{{\cal HLF}}

\renewcommand{\L}{{\mathscr L}}

\newcommand{\R}{\mathbb{R}}

\newcommand{\N}{\mathbb{N}}
\newcommand{\FF}{\mathbb{F}}

\newcommand{\K}{\mathbb{K}}

\newcommand{\modd}{{\rm mod}}

\newcommand{\Supp}{{\rm supp\;}}

\renewcommand{\span}{{\rm Span\ }}

\renewcommand{\modd}{\ {\rm mod}\ }

\newcommand{\size}{{\rm size}}

\newcommand{\bfa}{{\boldsymbol a}}
\newcommand{\bfb}{{\boldsymbol b}}
\newcommand{\bfc}{{\boldsymbol c}}

\newcommand{\bfi}{{\boldsymbol i}}

\newcommand{\bfl}{{\boldsymbol l}}

\newcommand{\bfx}{{\boldsymbol x}}
\newcommand{\bfy}{{\boldsymbol y}}
\newcommand{\bfz}{{\boldsymbol z}}

\newcommand{\bflambda}{{\boldsymbol\lambda}}
\newcommand{\bfalpha}{{\boldsymbol\alpha}}
\newcommand{\bfell}{{\boldsymbol l}}
\newcommand{\bell}{{\boldsymbol l}}

\newcommand{\bfepsilon}{{\boldsymbol\epsilon}}
\newcommand{\blambda}{{\boldsymbol\lambda}}

\renewcommand{\Pr}{{\bf Pr}}
\newcommand{\poly}{{\rm poly}}

\newtheorem{theorem}{Theorem}
\newtheorem{lemma}[theorem]{Lemma}

\newtheorem{claim}[theorem]{Claim}

\newtheorem{corollary}[theorem]{Corollary}

\newtheorem{proposition}[theorem]{Proposition}

\renewcommand{\Pr}{{\bf Pr}}
\newcommand{\ignore}[1]{{ }}

\reversemarginpar

\title{\bf Dense Testers:
Almost Linear Time \\ and Locally Explicit Constructions}

\author{
Nader H. Bshouty \\
Department of Computer Science\\
Technion, Israel \\
\texttt{bshouty@cs.technion.ac.il} \\
}

\begin{document}
\maketitle
\begin{abstract}
We develop a new notion called {\it $(1-\epsilon)$-tester for a
set $\cM$ of functions}  $f:\cA\to \cC$.  A $(1-\epsilon)$-tester
for $\cM$ maps each element $\bfa\in\cA$ to a finite number of
elements $B_\bfa=\{\bfb_1,\ldots,\bfb_t\}\subset\cB$ in a smaller
sub-domain $\cB\subset \cA$ where for every $f\in \cM$ if
$f(\bfa)\not=0$ then $f(\bfb)\not=0$ for at least $(1-\epsilon)$
fraction of the elements $\bfb$ of $B_\bfa$. I.e., if
$f(\bfa)\not=0$ then $\Pr_{\bfb\in B_\bfa}[f(\bfb)\not=0]\ge
1-\epsilon$. The {\it size} of the $(1-\epsilon)$-tester is
$\max_{\bfa\in\cA}|B_\bfa|$ and the goal is to minimize this size,
construct $B_\bfa$ in deterministic almost linear time and access
and compute each map in poly-log time.

We use tools from elementary algebra and algebraic function fields
to build $(1-\epsilon)$-testers of small size in deterministic
almost linear time. We also show that our constructions
are locally explicit, i.e.,
one can find any entry in the construction in time poly-log in the size
of the construction and the field size.
We also prove lower bounds that show that the
sizes of our testers and the densities are almost optimal.

Testers were used in [Bshouty, Testers and its application,  ITCS 2014] to construct almost optimal perfect hash families, universal sets, cover-free families,
separating hash functions, black box identity testing and hitting sets.
The dense testers in this paper shows that such constructions can be done
in almost linear time, are locally explicit and can be made to be dense.
\end{abstract}

\newpage
\tableofcontents \newpage

{
\section{Introduction}
A $(1-\epsilon)$-{\it tester} of a class of multivariate polynomials $\cM$ over
$n$ variables is a set $L$ of maps from a ``complex'' (algebraic)
structure $\cA^n$ (such as algebra over a field, algebraic function field, modules) to a ``simple'' algebraic structure (such as field or ring) $\cB^n$ that for every $f\in\cM$ preserve the property $f(\bfa)\not=0$ for at least
$(1-\epsilon)$ fraction of the maps, i.e., for all $f\in
\cM$ and $\bfa\in\cA^n$ if $f(\bfa)\not=0$ then
$f(\ell(\bfa))\not=0$ for at least $(1-\epsilon)$
fraction of the maps $\ell\in L$. See a formal
definition in Section~\ref{Tester}.

In this
paper we study $(1-\epsilon)$-testers when $\cA$, the domain of the functions in
$\cM$, is a field and $\cB\subset \cA$ is
a small subfield. We use tools from elementary algebra and
algebraic function fields to construct testers of almost optimal
size $|L|$ in almost linear time.
\ignore{In this paper,
the time of the construction includes the time of
constructing the tester $L$ and
the worst-case time complexity, over $i$, of computing the $i$th entry $\ell(\bfa)_i$ for all $\ell\in L$.
\footnote{Obviously, the complexity of computing $\ell(\bfa)$ is at most $n$
times the worst-time complexity, over $i$, of computing
the $i$th entry $\ell(\bfa)_i$}}

A construction is {\it globally explicit} if it runs in deterministic
polynomial time in the size of the construction and poly-log in the size of the field.
A {\it locally explicit construction} is
a construction where one can find any entry in the construction
in deterministic poly-log time in the size of the construction and the size of the field.
In particular, a locally explicit construction is also globally explicit.
The constructions in this paper are locally explicit constructions
and runs in almost linear time in the size of the construction.

We also give lower bounds that show that the size of our constructions
and their densities are almost optimal.

One application of $(1-\epsilon)$-testers is the following: Suppose we need to
construct a small set of vectors $S\subset\Sigma^n$ for some
alphabet $\Sigma$ that at least $(1-\epsilon)$ fraction
of its elements satisfy some property $P$. We map $\Sigma$
into a field $\FF$ and find a set of functions $\cM_P$ where
$S\subset\FF^n$ satisfies property $P$ if and only if $S$ is a
hitting set for $\cM_P$, i.e., for every $f\in \cM_P$ there is
$\bfa\in S$ such that $f(\bfa)\not=0$. We then extend $\FF$ to a
larger field $\K$ (or $\FF$-algebra $\cA$). Find $S'\subset \K^n$
that is a hitting set of density $(1-\epsilon_1)$ for $\cM_P$ (which supposed to be easier).
Then use $(1-(\epsilon-\epsilon_1))$-tester to change the hitting set $S'\subset \K^n$ over
$\K$ to a hitting set $S\subset \FF^n$ over $\FF$ of density $(1-\epsilon)$.

Non-dense Testers were first studied in \cite{B1}. They were used to
give a polynomial time constructions of almost
optimal perfect hash families, universal sets, cover-free families,
separating hash functions, black box identity testing and hitting sets.
Dense Testers were first mentioned in~\cite{B1} (see section 7 conclusion and
future work) where the application for new pseudorandom generators are
also mentioned as one of our future work. In \cite{GX14}, Guruswami and Xing, independently, used the same technique for similar construction.
The results in this paper show that all the
constructions in~\cite{B1} can be constructed
in almost linear time, are locally explicit and can be changed to be dense.

In this paper we consider two main classes of multivariate
polynomials over finite fields $\FF_q$ with $q$ elements. The
first class is $\P(\FF_q,n,d)$, the class of all multivariate
polynomials with $n$ variables and total degree~$d$. The second
class is $\HLF(\FF_q,n,d)$, the class of multilinear forms of degree $d$.
That is, the set of all multivariate
polynomials $f$ with $dn$ variables $x_{i,j}$, $i=1,\ldots,d$,
$j=1,\ldots,n$ where each monomial in $f$ is of the form
$x_{1,i_1}x_{2,i_2}\cdots x_{d,i_d}$. All the constructions in~\cite{B1}
are based on testers for the above two classes.

In Section~\ref{s01} we give some preliminary results. In Section~\ref{s08}
we give the definition of dense tester and prove some preliminary results
for dense testers.
In Section~\ref{s12} we give lower bounds for the size of dense testers
and for their density. In Section~\ref{s15} we give the (non-polynomial time)
constructions of dense testers. The almost linear time locally explicit
constructions are given in Section~\ref{s18}.
In Subsection~\ref{s19} and \ref{s20} we give constructions of dense testers
for $\P(\FF_q,n,d)$, $q\ge d+1$, from $\FF_{q^t}$ to $\FF_q$ with optimal density
of size within a factor of $poly(d/\epsilon)$ of the optimal size. In~\cite{B1}
we show that no such tester exists when $q\le d$. In Subsection~\ref{s21} we give
constructions of dense testers
for $\HLF(\FF_q,n,d)$, from $\FF_{q^t}$ to $\FF_q$ for any $q$.

\section{Preliminary Definitions and Results}\label{s01}
In this section we give some definitions and results from the literature
that will be used throughout the paper
\subsection{Multivariate Polynomial}\label{Multivariate_Polynomial}\label{s02}

In this section we define the set of multivariate polynomials over
a field $\FF$.

Let $\FF$ be a field and $\bfx=(x_1,\ldots,x_n)$ be indeterminates
(or variables) over the field $\FF$. The ring of {\it multivariate
polynomials} in the indeterminates $x_1,\ldots,x_n$ over $\FF$ is
$\FF[x_1,\ldots,x_n]$ (or $\FF[\bfx]$). Let
$\bfi=(i_1,\ldots,i_n)\in \N^n$. We denote
by~{$\bfx^{\bfi}$}\label{label-bfxhbfi} the {\it {monomial}{}}
$x_1^{i_1}\cdots x_n^{i_n}$. Every multivariate polynomial $f$ in
$\FF[\bfx]$ can be represented as
\begin{eqnarray}\label{multipoly}
f(\bfx)=\sum_{\bfi\in I} a_{\bfi}\bfx^\bfi
\end{eqnarray}
for some finite set $I \subset \N^n$ and $a_\bfi\in \FF\backslash
\{0\}$ for all $\bfi\in I$.

When the field $\FF$ is infinite,  the representation in
(\ref{multipoly}) is unique. Not every function $f':\FF^n\to\FF$
can be represented as multivariate polynomial. Take for example a
function $f'(x_1)$ with one variable that has infinite number of
roots.

When the field $\FF$ is finite, then using, for example, Lagrange
interpolation, every function $f':\FF^n\to\FF$ can be represented
as multivariate polynomial $f\in \FF[\bfx]$. There may be many
representations for the same function $f':\FF^n\to\FF$ but a
unique one that satisfies $I\subseteq \{0,1,\ldots,|\FF|-1\}^n$.
This follows from the fact that $x^{|\FF|}=x$ in $\FF$. We denote
this unique representation by $R(f')$ and denote $f'$ by $F(f)$.
In this paper, functions and their representations in $\FF[\bfx]$
are used exchangeably. So by $R(f)$ we mean $R(F(f))$.

For a monomial $M$ when we say that $M$ is a {{\it monomial in}
$f$}{} we mean that $R(M)$ is a monomial that appears in $R(f)$.
The constant $a_\bfi\in\FF\backslash\{0\}$ in (\ref{multipoly}) is
called the {\it coefficient} of the monomial~$\bfx^{\bfi}$ in $f$
and it is the coefficient of $R(\bfx^\bfi)$ in $R(f)$. When
$\bfx^\bfi$ is not a monomial in $f$ then we say that its
coefficient is $0$.

The {\it degree}, {$\deg(M)$}, of a monomial $M={\bfx}^{\bfi}$ is
$i_1+i_2+\cdots+i_n$. The {\it degree of $x_j$ in} $M$,
{$\deg_{x_j}(M)$} is $i_j$. Therefore,
$$\deg(M)=\sum_{i=1}^n \deg_{x_i}(M).$$
Let $f\in \FF[\bfx]$ and let $g=R(f)$. The {\it degree} (or {\it
total degree}) $\deg(f)$ is the maximum degree of the monomials in
$g$. The degree of $x_i$ in $f$, {$\deg_{x_i}(f)$}, is the maximum
degree of $x_i$ in the monomials in~$g$, i.e., the degree of $g$
when written as a univariate polynomial in the variable $x_i$. The
{\it variable degree} of $f$ is the maximum over the degree of
each variable in $f$, i.e., $\max_i\deg_{x_i}(f)$.
The size of $f$, $\size(f)$, is the number of monomials in $g$.

\subsubsection{Classes of Multivariate
Polynomials}\label{Classes-of-Multivariate-Polynomials}\label{s03}

In this section we define classes of multivariate polynomials that
will be studied in the sequel.

We first define
\begin{enumerate}
\item $\P(\FF,n)$ is the class of all multivariate polynomials in
$\FF[x_1,\ldots,x_n]$ of variable degree at most $|\FF|-1$. When
$\FF$ is finite, every functions $f:\FF^n\to \FF$ can be
represented by some multivariate polynomial in $\P(\FF,n)$. When
$\FF$ is infinite $\P(\FF,n)=\FF[x_1,\ldots,x_n]$.

\item $\P(\FF,n,(d,r))$ is the class of all multivariate
polynomials in $\P(\FF,n)$ of degree at most~$d$ and variable
degree at most~$r$.

\item $\P(\FF,n,d)=\P(\FF,n,(d,|\FF|-1))$ is the class of all
multivariate polynomials in $\P(\FF,n)$ of degree at most $d$.

\item {${\it \cal HP}(\FF,n)$} is the class of all homogeneous
polynomials in $\P(\FF,n)$. A multivariate polynomial is called
{\it homogeneous multivariate polynomial} if all its monomials
have the same degree. In the same way as
above one can define
$\HP(\FF,n,(d,r))$ and $\HP(\FF,n,d)$.

\end{enumerate}

\subsubsection{Multivariate Form}\label{Multivariate-Form}\label{s05}

Let $\bfy=(\bfy_1,\ldots,\bfy_m)$ where
$\bfy_i=(y_{i,1},\ldots,y_{i,n})$ are indeterminates over $\FF$
for $i=1,\ldots,m$. A {\it multivariate form} in $\bfy$ is a
multivariate polynomial in $\bfy$. That is, an element of
$$\FF[y_{1,1},\ldots,y_{1,n},\ldots,y_{m,1},\ldots,y_{m,n}].$$ We
denote this class by $\FF[\bfy]$ or $\FF[\bfy_1,\ldots,\bfy_m]$.
Let $\HLF(\FF,n,m)$ be
the class of all multilinear forms $f$ over
$\bfy=(\bfy_1,\ldots,\bfy_m)$ where each monomial in $f$ contains
exactly one variable from $\bfy_i$ for every $i$. In~\cite{B1},
polynomials in $\HLF(\FF,n,m)$ are called $(n,m)$-{\it multilinear
polynomials}. Notice that $\HLF(\FF,n,2)$ is the class of bilinear
forms $\bfy_1^TA\bfy_2$ where $A\in\FF^{n\times n}$.

\subsection{Algebraic Complexity}\label{s06}

In this section we give some known results in algebraic
complexity that will be used in the sequel

\subsubsection{Complexity of Constructing Irreducible Polynomials
and $\FF_{q^t}$}\label{Tpoly1}\label{s07}

In some applications the construction of irreducible polynomials
of degree $n$ over $\FF_q$ and the construction of the field
$\FF_{q^t}$ is also needed and their complexity must be included
in the overall time complexity of the problem.

To construct the field $\FF_{q^t}$ one should construct an
irreducible polynomial $f(x)$ of degree $t$ in $\FF_q[x]$ and then
use the representation $\FF_{q^t}=\FF_q[x]/(f(x))$. For a
comprehensive survey on this problem see \cite{S99} Chapter~3. See
also \cite{AL86,E89,S90}. We give here the results that will be
used in this paper.

\begin{lemma} \labell{FF1}
Let $\FF_q$ be a field of characteristic $p$. There is an
algorithm that constructs an irreducible polynomial of degree $t$
with $T$ arithmetic operations in the field $\FF_q$ where $T$ is
as described in the following table.

\center{
\begin{tabular}{|c|c|c|l|l|}
\hline
{\bf Type} & {\bf Field} & {\bf Assumption} &Time $=T$ & $T=${\bf $\tilde O$}\\
\hline\hline
Probabilistic & Any & $-$& $O\left(t^2\log^{2+\epsilon} t+t\log q\log^{1+\epsilon} t\right)$ & $\tilde O(t^2)$\\
\hline
Deterministic & Any & $-$& $O\left(p^{1/2+\epsilon}t^{3+\epsilon}+
(\log q)^{2+\epsilon}t^{4+\epsilon}\right)$ & $\tilde O(p^{1/2}t^3+t^4)$\\
\hline
Deterministic & Any & ERH &$O(\log^2q+t^{4+\epsilon}\log q)$  & $\tilde O(t^4)$\\
\hline
Deterministic & $\FF_2$ & $-$& $O(t^{3+\epsilon})$ & $\tilde O(t^3)$\\
\hline
\end{tabular}}

Here $ERH$ stands for the Extended Riemann Hypothesis and
$\epsilon$ is any small constant.
\end{lemma}

Here $\tilde O(M)$ means $O(M\cdot t^\epsilon\cdot poly(\log q))$.
In the sequel when we give a complexity for constructing
a field or irreducible polynomial then $\tilde O(M)$ means
$O(M\cdot t^\epsilon \cdot poly(\log M,\log q))$ but for all the constructions in this paper
$\tilde O(M)$ will mean $O(M\cdot  poly(\log M,\log q))$.

In Lemma~\ref{prerec} one should construct many irreducible
polynomials of certain degree. We now prove the following result

\begin{lemma}\labell{ManI} There is a deterministic algorithm that runs in time
$$\tilde O(mt+t^3 p^{1/2}+t^4)$$
(and $\tilde O(mt+t^4)$ assuming ERH) and construct $m$ distinct irreducible polynomials of degree $t$ in $\FF_q[x]$ and their roots.
\end{lemma}
\begin{proof}
By Lemma~\ref{FF1}, $\FF_{q^t}$ can be constructed in
time $O\left(t^{3+\epsilon}p^{1/2+\epsilon}+
(\log q)^{2+\epsilon}t^{4+\epsilon}\right)$. It is known that a normal basis
$\{\alpha,\alpha^q,\alpha^{q^2},\ldots,\alpha^{q^{t-1}}\}$ in
$\FF_{q^t}$ can be constructed in time $O(t^3+t\log t\log\log t\log q)$,
\cite{L91,P94}. For any
$\bflambda=(\lambda_1,\lambda_2,\ldots,\lambda_t)\in \FF_q^t$, the element
$$\beta_\bflambda:=\lambda_1\alpha+\lambda_2\alpha^{q}+\lambda_3 \alpha^{q}+\cdots+\lambda_{t-1}\alpha^{q^{t-1}}$$ is a root of an irreducible polynomial of degree $t$ if and only if $\beta_\bflambda,\beta_\bflambda^q,\beta_\bflambda^{q^2},\ldots,\beta_\bflambda^{q^{t-1}}$ are distinct. It is easy to see that this is true if and only if the vectors
$$\bflambda^0:=\bflambda,\ \bflambda^1:=(\lambda_t,\lambda_1,\ldots,\lambda_{t-1}), \ \bflambda^2:=(\lambda_{t-1},\lambda_t,\lambda_1,\ldots,\lambda_{t-2}),\cdots,
\bflambda^{t-1}:=(\lambda_2,\lambda_3,\ldots,\lambda_{t},\lambda_1)$$ are distinct. Such $\bflambda$ is called a vector of period $t$.

If we have a vector $\bflambda$ of period $t$ then $\beta_\bflambda$ is a root of irreducible polynomial $f_{\beta_\bflambda}(x)$ of degree $t$ where
$f_{\beta_\bflambda}(x)\equiv (x-\beta_\bflambda)(x-\beta_\bflambda^q)\cdots (x-\beta_\bflambda^{q^{t-1}})$. The coefficients of the polynomial $f_{\beta_\bflambda}(x)$ can be computed in time $O(t\log^2 t\log\log t)$. See Theorem A in \cite{S99} and references within. Therefore, it remains to construct $m$ vectors of period $t$.

Now choose any total order $<$  on $\FF_q$ and consider the lexicographic order in $\FF_q^t$ with respect to $<$ and consider the sequence of all the elements of $\FF_q^t$ with this order. It is easy to see that for any two consecutive elements $\bflambda_1,\bflambda_2\in \FF_q^t$ in this sequence there is at least one $\bflambda_i$, $i\in \{1,2\}$ of period $t$. Also, each irreducible polynomial $f_{\beta_\bflambda}$ of degree $t$ can be constructed by exactly $t$ elements (i.e., $\bflambda^0,\bflambda^1,\ldots,\bflambda^{t-1}$) in the sequence. This implies that the first $2tm$ elements in this sequence generate at least $m$ distinct irreducible polynomials.
\end{proof}

The following result will be used for the local explicit
constructions and is proved in Appendix A.
\begin{lemma}\labell{ManIle} Let $r= \lfloor q^{t-2}/2t\rfloor$. There is a total order
on a set of $r$ irreducible polynomials of degree $t$ in $\FF_q[x]$ and a
deterministic algorithm that with an input $m$ runs in time
$$\tilde O(t^3 p^{1/2}+t^4)$$ and constructs the $m$th irreducible polynomial in that order
with its roots.

The time is $\tilde O(t^4)$ assuming ERH.
\end{lemma}

Throughout this paper, the complexities are given without the assumption of ERH. When
ERH is assumed then just drop the $p^{1/2}$ from the complexities.
}

\section{Dense
Tester} \label{Testerr}\labell{s08}

In this section we define $(1-\epsilon)$-testers and give some
preliminary results.

\subsection{Definition of $(1-\epsilon)$-Tester}\labell{TesterS}\labell{s09}

In this section we define $(1-\epsilon)$-tester. We will assume
that all the $\FF$-algebras in this paper are commutative,
although most of the results are also true for noncommutative
$\FF$-algebras.

Let $\FF$ be a field and $\cA$ and $\cB$ be two $\FF$-algebras.
Let $0\le \epsilon< 1$ and $\bepsilon=1-\epsilon$. Let ${\cal
M}\subseteq \FF[x_1,x_2,\ldots,x_n]$ be a class of multivariate
polynomial. Let $S\subseteq \cA$ and $R\subseteq \cB$ be linear
subspaces over $\FF$ and $L=\{\bfl_1,\ldots,\bfl_\nu\}$ be a set
of (not necessarily linear) maps $\bfl_i:S^n\to R^n$,
$i=1,\ldots,\nu$. We say that $L$ is $({\cal
M},S,R)$-$\bepsilon$-{\it tester} if for every
$\bfa=(a_1,\ldots,a_n)\in S^n$ and $f\in{\cal M}$ we have
$$
f(\bfa)\not=0 \ \Longrightarrow\ \Pr_{\bfl\in
L}[f(\bfl(\bfa))\not=0]\ge \bepsilon$$
where the probability is uniform over the choices of $\bfl\in L$.

The integer $\nu=|L|$ is called the {\it size of the
$\bepsilon$-tester}. The minimum size of such tester is denoted
by $\nu^\circ_{R}(\cM,S,\bepsilon)$. If no such tester exists then
we write $\nu^\circ_{R}(\cM,S,\bepsilon)=\infty$. When $S$ and $R$
are known from the context we then just say that $L$ is
$\bepsilon${\it -tester for} $\cM$.

An {\it $(\cM,S,R)$-tester} is an $(\cM,S,R)$-$\bepsilon$-tester
for some $\epsilon<1$. Tester was studied in \cite{B1}. The
minimum size of an $(\cM,S,R)$-tester is denoted by
$\nu^\circ_{R}(\cM,S)$. Obviously we have
\begin{eqnarray}\labell{Tester}
\nu^\circ_{R}\left(\cM,S,\frac{1}{\nu^\circ_{R}(\cM,S)}\right)=
\nu^\circ_{R}(\cM,S).
\end{eqnarray}
Obviously, $L$ is an $({\cal M},S,R)$-$\bepsilon$-{\it tester} if
and only if for every $L'\subseteq L$ where $|L'|=\lfloor \epsilon
|L|\rfloor+1$, $L'$ is $({\cal M},S,R)$-{\it tester}.

We say that the $\bepsilon$-tester $L$ is {\it componentwise} if
for every $\bfl_i\in L$ we
have $\bfl_i(\bfa)=(l_{i,1}(a_1),\ldots,l_{i,n}(a_n))$ for some
$l_{i,j}:S\to R$. A componentwise tester is called {\it linear} if
each $l_{i,j}$ is a linear map and is called {\it reducible} if
$\cA$ and $\cB$ has identity elements $1_\cA$ and $1_\cB$,
respectively, $1_\cA\in S$ and $l_{i,j}(1_\cA)=1_\cB$ for all
$l_{i,j}$.

We will also allow $L=\{l_1,\ldots,l_\nu\}$ to be a set of maps
$l_i:S\to R$, for $i=1,\ldots,\nu$ (rather than maps $S^n\to R^n$). In that case $\bfl_i:S^n\to
R^n$ is defined as $\bfl_i(\bfa)=(l_i(a_1),\ldots,l_i(a_n))$ where
$\bfa=(a_1,\ldots,a_n)\in S^n$. In such case we call the
$\bepsilon$-tester a {\it symmetric $\bepsilon$-tester}.

In this paper we will mainly study $\bepsilon$-testers for the
class of multilinear forms of degree $d$ and multivariate
polynomials of degree $d$.

We will use the following abbreviations

\begin{center}
\begin{tabular}{|l|l|l|}
The Expression & Abbreviation & or the Abbreviation\\
\hline $\nu^\circ_R(\P(\FF,n,d),S,\bepsilon)$ & $\nu_R^\P(d,S,\bepsilon)$ & $\nu_R^\P((d,\FF),S,\bepsilon)$\\
\hline $\nu^\circ_R(\HP(\FF,n,d),S,\bepsilon)$ & $\nu_R^\HP(d,S,\bepsilon)$ & $\nu_R^\HP((d,\FF),S,\bepsilon)$\\
\hline $\nu^\circ_R(\HLF(\FF,n,m),S,\bepsilon)$ & $\nu_R(m,S,\bepsilon)$ & $\nu_R((m,\FF),S,\bepsilon)$\\
\hline
\end{tabular}
\end{center}
In the abbreviations $\nu_R^\P(d,S,\bepsilon)$, (respectively,
$\nu_R^\HP(d,S,\bepsilon)$ and $\nu_R(m,S,\bepsilon)$)
we assume that the ground field $\FF$ is known from the context,
e.g., when $R=\FF$. Otherwise, we write
$\nu_R^\P((d,\FF),S,\bepsilon)$, (respectively, $\nu_R^\HP((d,\FF),S,\bepsilon)$ and
$\nu_R((m,\FF),S,\bepsilon)$)

Notice that we omitted the parameter $n$ from the abbreviation.
This is because, for the classes we will study here, the value of
$\nu_R^\circ$ is monotone non-decreasing in $n$ and we are
interested in the worst case size of such testers. So one can define
$\nu_R^\P(d,S,\bepsilon)=\lim_{n\to \infty}
\nu_R^\P(\P(\FF,n,d),S,\bepsilon).$

\ignore{We will add the letters $c,l,r$ and $s$ as superscript
when we refer to componentwise, linear, reducible and symmetric
testers. For example $\nu_R^{rs,\P}(d,S,\bepsilon)$ is the minimum
size reducible and symmetric
$(\P(\FF,n,d),S,R)$-$\bepsilon$-tester.}

\subsection{Preliminary Results for Testers}\labell{s10}

In this section we prove some preliminary results on
$\bepsilon$-testers that will be frequently used in the sequel.

The first two Lemmas follows from the definition of
$\bepsilon$-tester

\begin{lemma} \labell{Triv} Let $\cA$ and $\cB$ be commutative
$\FF$-algebras. Let $S_1\subseteq S_2\subseteq \cA$, $R_2\subseteq
R_1\subseteq \cB$ be linear subspaces over $\FF$,
$\cN\subseteq\cM\subseteq \FF[x_1,\ldots,x_n]$ and
$\epsilon_1\ge\epsilon_2$. If $L$ is
$(\cM,S_2,R_2)$-$\bepsilon_2$-tester then it is
$(\cN,S_1,R_1)$-$\bepsilon_1$-tester. In particular,
$$\nu_{R_1}^\circ(\cN,S_1,\bepsilon_1)\le
\nu_{R_2}^\circ(\cM,S_2,\bepsilon_2).$$
\end{lemma}

\begin{lemma} \labell{ctb1} Let $\cA,\cB$ and $\cC$ be commutative
$\FF$-algebras. Let $S_1\subseteq  \cA,\ S_2\subseteq \cB$ and
$S_3\subseteq \cC$ be linear subspaces over $\FF$ and
$\cM\subseteq \FF[x_1,\ldots,x_n]$. If $L_1$ is a
$(\cM,S_1,S_2)$-$\bepsilon_1$-tester and $L_2$ is a
$(\cM,S_2,S_3)$-$\bepsilon_2$-tester then $L_2\circ
L_1:=\{\bfl_2(\bfl_1)\ |\ \bfl_1\in L_1, \bfl_2\in L_2\}$ is
$(\cM,S_1,S_3)$-$(\bepsilon_1\bepsilon_2)$-tester. In particular,
$$\nu_{S_3}^\circ(\cM,S_1,\bepsilon_1\bepsilon_2)\le
\nu_{S_3}^\circ(\cM,S_2,\bepsilon_1)\cdot
\nu_{S_2}^\circ(\cM,S_1,\bepsilon_2).$$
\end{lemma}

In particular we have
\begin{corollary} \labell{ctb11} Let $\K$ be an extension field of $\FF$
and $\cA$ be a $\K$-algebra. Let $\cM\subseteq
\FF[x_1,\ldots,x_n]$. Then
$$\nu_{\FF}^\circ(\cM,\cA,\bepsilon_1\bepsilon_2)\le
\nu_{\FF}^\circ(\cM,\K,\bepsilon_1)
\cdot\nu_{\K}^\circ(\cM,\cA,\bepsilon_2).$$

In particular, for any integers $m_1$ and $m_2$ we have
$$\nu_{\FF_q}^\circ(\cM,\FF_{q^{m_1m_2}},\bepsilon_1\bepsilon_2)
\le \nu_{\FF_q}^\circ(\cM,\FF_{q^{m_1}},\bepsilon_1)
\cdot\nu_{\FF_{q^{m_1}}}^\circ(\cM,\FF_{q^{m_1m_2}},\bepsilon_2).$$
\end{corollary}

The above results are also true for componentwise, linear,
reducible (assuming $1$ is in all the sets) and symmetric
$\bepsilon$-testers. We state this in the following
\begin{lemma} \label{CLRS} The results in Lemma~\ref{Triv},
Lemma~\ref{ctb1} and Corollary~\ref{ctb11} are also true for
componentwise, linear, reducible and symmetric $\bepsilon$-tester.
\end{lemma}
Since ${\epsilon_1+\epsilon_2}\ge
\overline{\bepsilon_1\bepsilon_2}$, by
Lemma~\ref{Triv},~\ref{ctb1} and Corollary~\ref{ctb11} we also
have
\begin{eqnarray}\labell{ctb1pl}\nu_{S_3}^\circ(\cM,S_1,\overline{\epsilon_1+\epsilon_2})\le
\nu_{S_3}^\circ(\cM,S_2,\bepsilon_1)\cdot
\nu_{S_2}^\circ(\cM,S_1,\bepsilon_2),\end{eqnarray}
\begin{eqnarray}\labell{ctb11pl}\nu_{\FF}^\circ(\cM,\cA,\overline{\epsilon_1+\epsilon_2})\le
\nu_{\FF}^\circ(\cM,\K,\bepsilon_1)
\cdot\nu_{\K}^\circ(\cM,\cA,\bepsilon_2)\end{eqnarray} and
\begin{eqnarray}\labell{ctb11pl2}\nu_{\FF_q}^\circ(\cM,\FF_{q^{m_1m_2}},\overline{\epsilon_1+\epsilon_2})
\le \nu_{\FF_q}^\circ(\cM,\FF_{q^{m_1}},\bepsilon_1)
\cdot\nu_{\FF_{q^{m_1}}}^\circ(\cM,\FF_{q^{m_1m_2}},\bepsilon_2).\end{eqnarray}

We now prove

\begin{lemma} \labell{ctb2} Let $\cA$ be a commutative
$\FF$-algebra and $S\subseteq \cA$ be a linear subspace over
$\FF$. Let $$\cM\subseteq
\FF[\bfx]\FF[\bfy]:=\left\{\left.\sum_{i=1}^s h_i(\bfx)g_i(\bfy)\
\right|\ h_i\in \FF[\bfx], g_i\in\FF[\bfy], s\in\N\right\}$$ be a
set of multivariate polynomials where $\bfx=(x_1,\ldots,x_n)$ and
$\bfy=(y_1,\ldots,y_m)$ are distinct indeterminates. Let
$$\cM_\bfx=\left\{\sum_{i=1}^s\lambda_ih_i(\bfx)\
\left|\begin{array}{c}\;\\ \;\end{array}\right.\
\sum_{i=1}^sh_i(\bfx)g_i(\bfy)\in \cM,\ \bflambda\in \FF^s,s\in \N
\right\}$$ and
$$\cM_\bfy=\left\{\sum_{i=1}^s\lambda_ig_i(\bfy)\
\left|\begin{array}{c}\;\\ \;\end{array}\right.\
\sum_{i=1}^sh_i(\bfx)g_i(\bfy)\in \cM,\ \bflambda\in
\FF^s,s\in\N\right\}.$$ If $L_\bfx$ is a
$(\cM_\bfx,S,\FF)$-$\bepsilon_\bfx$-tester and $L_\bfy$ is a
$(\cM_\bfy,S,\FF)$-$\bepsilon_\bfy$-tester then $L_\bfx\times
L_\bfy$ is a
$(\cM,S,\FF)$-$(\bepsilon_\bfx\bepsilon_\bfy)$-tester. In
particular,
$$\nu_{\FF}^\circ(\cM,S,\bepsilon_\bfx\bepsilon_\bfy)
\le \nu_{\FF}^\circ(\cM_\bfx,S,\bepsilon_\bfx)\cdot
\nu_\FF^\circ(\cM_\bfy,S,\bepsilon_\bfy).$$
\end{lemma}
\begin{proof} Suppose for some $f(\bfx,\bfy)=\sum_{i=1}^s
h_i(\bfx)g_i(\bfy)\in \cM$ and $(\bfa,\bfb)\in S^{n+m}$ we have
$$\Pr_{(\bell_\bfx,\bell_\bfy)\in L_\bfx\times L_\bfy}
\left[f(\bell_\bfx(\bfa),\bell_\bfy(\bfb))\not=0\right]<
\bepsilon_\bfx\bepsilon_\bfy.$$ By Markov bound we have that more
than $\epsilon_\bfx |L_\bfx|$ of the elements $\bell_\bfx\in
L_\bfx$ satisfies
$$\Pr_{\bell_\bfy\in L_\bfy}\left[f(\bell_\bfx(\bfa),
\bell_\bfy(\bfb))\not=0\right]< \bepsilon_\bfy.$$

Since $f(\bfell_\bfx(\bfa),\bfy)\in \cM_\bfy$ and $L_\bfy$ is an
$(\cM_\bfy,S,\FF)$-$\bepsilon_\bfy$-tester it follows that for
more than $\epsilon_\bfx |L_\bfx|$ of the elements $\bfell_\bfx\in
L_\bfx$ we have $f(\bfell_\bfx(\bfa),\bfb)=0$. Let $\ell$ be any
linear map in $\cA^*$. Then for more than $\epsilon_\bfx |L_\bfx|$
of the elements $\bfell_\bfx\in L_\bfx$ we have
$$
\sum_{i=1}^sh_i(\bfell_\bfx(\bfa))\ell(g_i(\bfb))=
\ell(f(\bfell_\bfx(\bfa)),\bfb))=0.$$

Since $\sum_{i=1}^sh_i(\bfx)\ell(g_i(\bfb))\in \cM_\bfx$ and
$L_\bfx$ is an $(\cM_\bfx,S,\FF)$-$\bepsilon_\bfx$-tester we have
$\sum_{i=1}^sh_i(\bfa)\ell(g_i(\bfb))=0$. Notice that this is true
for any linear map~$\ell\in \cA^*$. Now let
$\{\omega_1,\ldots,\omega_r\}\subset \cA$ be a basis for
$\span_\FF\{g_1(\bfb),\ldots,g_s(\bfb)\}$, the linear subspace
spanned by $\{g_1(\bfb),\ldots,g_s(\bfb)\}$ over $\FF$. Let
$\ell_{\omega_i}$, $i=1,\ldots,s$, be linear maps in $\cA^*$ such
that $g_i(\bfb)=\sum_{j=1}^r\ell_{\omega_j}(g_i(\bfb))\omega_j$.
Then
\begin{eqnarray*}
f(\bfa,\bfb)&=&\sum_{i=1}^sh_i(\bfa)g_i(\bfb)\\
&=&
\sum_{i=1}^sh_i(\bfa)\sum_{j=1}^r\ell_{\omega_j}(g_i(\bfb))\omega_j\\
&=&\sum_{j=1}^r\omega_j\sum_{i=1}^sh_i(\bfa)\ell_{\omega_j}(g_i(\bfb))=0.
\end{eqnarray*}
\end{proof}

\begin{lemma} \label{psym} Lemma~\ref{ctb2} is also true for componentwise,
linear and reducible $\bepsilon$-testers and not necessarily true
for symmetric $\bepsilon$-testers.
\end{lemma}

\ignore{In the sequel we give an example of a class $\cM$ where $\cM_\bfx$
and $\cM_\bfy$ have symmetric $\bepsilon$-testers and $\cM$ has no
symmetric $\bepsilon$-testers.}

For an indeterminate $X$ over $\FF$ and an integer $k\ge 1$, let
$\FF[X]_{k}$ be the linear space of all polynomials in $\FF[X]$ of
degree at most $k$.

\begin{lemma}\labell{FFxFFal}
Let $\K/\FF$ be a field extension and $\alpha\in \K$ algebraic
over $\FF$ of degree $t$. Let $\bfx=(x_1,\ldots,x_n)$ and $X$ be
indeterminates over $\FF$ and $\cM\subseteq \FF[\bfx]$. There is a
linear symmetric reducible $(\cM,\FF(\alpha)$
$,\FF[X]_{t-1})$-$1$-tester of size $1$. In particular,
$$\nu^\circ_{\FF_q[X]_{t-1}}(\cM,\FF_{q^t},1)=1.$$
\end{lemma}
\begin{proof}
Every element in $\FF(\alpha)$ can be written as
$\omega_0+\omega_1\alpha+\cdots+\omega_{t-1}\alpha^{t-1}$ where
$\omega_i\in\FF$ for $i=0,1,\ldots,t-1$. Define the map
$l_X:\FF(\alpha)\to \FF[X]_{t-1}$,
$l_X(\omega_0+\omega_1\alpha+\cdots+\omega_{t-1}\alpha^{t-1})=
\omega_0+\omega_1X+\cdots+\omega_{t-1}X^{t-1}$. Notice that for
$a\in \FF(\alpha)$, $l_X(a)|_{X\gets\alpha}=a$. Therefore, for
$\bfa=(a_1,\ldots,a_n)\in \FF(\alpha)^n$ and $f\in \cM$ if
$f(l_X(a_1),\ldots,l_X(a_n))=0$ then
$$f(\bfa)=f(l_X(a_1)|_{X\gets\alpha},\ldots,l_X(a_n)|_{X\gets\alpha})=f(l_X(a_1),\ldots,l_X(a_n))|_{X\gets\alpha}=0.$$
This gives a symmetric $(\cM,\FF(\alpha),\FF[X]_{t-1})$-$1$-tester
of size~$1$. Since $\l_X$ is a linear map and $l_X(1)=1$ the
tester is also reducible.
\end{proof}

\subsection{Preliminary Results for Polynomials of Degree $d$}\labell{s11}

In this section we prove some results related to testers for $\cP(\FF_q,n, d)$, $\HP(\FF_q,n, d)$
and $\HLF(\FF_q,n,d)$. We remind the reader that
$\nu_R^\P(d,S,\bepsilon)=\nu^\circ_R(\P(\FF,n,d),S,\bepsilon)$,
$\nu_R^\HP(d,S,\bepsilon)=\nu^\circ_R(\HP(\FF,n,d),S,\bepsilon)$
and $\nu_R(m,S,\bepsilon)=\nu^\circ_R(\HLF(\FF,n,m)$ $,S,\bepsilon)$.

\noindent
{\bf Important Note 1:} Throughout this paper,
we will, without stating explicitly in the results,
identify every inequality in $\nu_R^\P, \nu_R^\HP$ or $\nu_R$ with its
corresponding construction and time complexity.
For example, when we write
$$\nu_{\FF}(d_1+d_2,S,\bepsilon_1\bepsilon_2)\le
\nu_{\FF}(d_1,S,\bepsilon_1)\cdot \nu_\FF(d_2,S,\bepsilon_2)$$
we also mean the following two statements:
\begin{enumerate}
\item From $(\HLF(\FF,n,d_1),$ $S,\FF)$-$\bepsilon_1$-tester of size $s_1$
and $(\HLF(\FF,n,d_2)$ $,S,\FF)$-$\bepsilon_2$-tester of size $s_2$ one can
construct in deterministic {\bf linear time} (if not explicitly stated otherwise) a
$(\HLF(\FF,n,$ $d_1+d_2),S,\FF)$-$\bepsilon_1\bepsilon_2$-tester of size $s_1s_2$.

\item If any entry of any map (i.e., $\bfl(\bfa)_i$ for any $\bfl\in L$ and any $\bfa\in S^n$) of the $(\HLF(\FF,n,d_1),$ $S,\FF)$-$\bepsilon_1$-tester can be
constructed and computed in time $T_1$ and any entry
of any map of the $(\HLF(\FF,n,d_2)$ $,S,\FF)$-$\bepsilon_2$-tester can be
constructed and computed in time $T_2$ then
any entry of any map of the $(\HLF(\FF,n,$ $d_1+d_2),S,\FF)$-$\bepsilon_1\bepsilon_2$-tester
can be constructed and computed in time $T_1+T_2+O(1)$.
\end{enumerate}

{\bf Important Note 2:} In this paper,
the time of the construction is the time of
constructing all the maps in the tester $L$. Denote this time by $T'$. 
The time of constructing
and computing any entry of any map is
the worst-case time complexity, over $i$ and all $\bfl\in L$, of computing the $i$th entry $\bfl(\bfa)_i$.
Denote this time by $T''$.
Obviously, the complexity of computing $\bfl(\bfa)$ is at most $nT''$
and the time of constructing and computing all the maps is less than $T'+|L|\cdot n T''$.

We also remind the reader that $\tilde O(M)$ means
$O(M\cdot poly(\log M,\log q))$. Here $poly(\log q)$ is added for
the complexity of the arithmetic computations in the ground field $\FF_q$.

First we prove
\begin{lemma}\labell{basn} We have
\begin{enumerate}
\item\labell{basn1} $\nu_R(d,S,\bepsilon)\le
\nu^\HP_R(d,S,\bepsilon)\le \nu^\cP_R(d,S,\bepsilon).$

\item\labell{basn12} $\nu_{\FF_q}^\P(d,\FF_{q^{m_1m_2}},\bepsilon_1\bepsilon_2)
\le \nu_{\FF_q}^\P(d,\FF_{q^{m_1}},\bepsilon_1)
\cdot\nu_{\FF_{q^{m_1}}}^\P(d,\FF_{q^{m_1m_2}},\bepsilon_2).$

\item\labell{basn13} $\nu_{\FF_q}(d,\FF_{q^{m_1m_2}},\bepsilon_1\bepsilon_2)
\le \nu_{\FF_q}(d,\FF_{q^{m_1}},\bepsilon_1)
\cdot\nu_{\FF_{q^{m_1}}}(d,\FF_{q^{m_1m_2}},\bepsilon_2).$

\item\labell{basn2}
$\nu_{\FF}(d_1+d_2,S,\bepsilon_1\bepsilon_2)\le
\nu_{\FF}(d_1,S,\bepsilon_1)\cdot \nu_\FF(d_2,S,\bepsilon_2).$
\end{enumerate}
\end{lemma}
\begin{proof} {\it \ref{basn1}} follows from Lemma~\ref{Triv}.
{\it \ref{basn12}} and {\it \ref{basn13}} follows from Corollary~\ref{ctb11}.
{\it\ref{basn2}} follows from Lemma~\ref{ctb2}.
\end{proof}

The following lemma gives an upper bound for the size of a dense tester when the ground field
$\FF_q$ is very large. In the sequel we show that this bound is tight.

\begin{lemma} \labell{qdt} We have
\begin{enumerate}
\item \label{HPre0} $\nu_{\FF_q}^\cP(d,\FF_{q^t},\bepsilon)\le
\nu_{\FF_q}^\P(d,\FF_q[X]_{t-1},\bepsilon)\mbox{\ and\ }
\nu_{\FF_q}^\HP(d,\FF_{q^t},\bepsilon)\le
\nu_{\FF_q}^\HP(d,\FF_q[X]_{t-1},\bepsilon).$

\item \label{HPre1} If $q\ge d(t-1)+1$ then for any $r$ such that
$q\ge r\ge d(t-1)+1$ and $\epsilon=d(t-1)/r$ we have
$$\nu_{\FF_q}^\cP(d,\FF_{q^t},\bepsilon)\le
\nu_{\FF_q}^\P(d,\FF_q[X]_{t-1},\bepsilon)\le
\frac{d(t-1)}{\epsilon}.$$

\item \label{HPre} If $q\ge d(t-1)$ then for any $r$ such that
$q+1\ge r\ge d(t-1)+1$ and $\epsilon=d(t-1)/r$ we have
$$\nu_{\FF_q}^\HP(d,\FF_{q^t},\bepsilon)\le
\nu_{\FF_q}^\HP(d,\FF_q[X]_{t-1},\bepsilon)\le
\frac{d(t-1)}{\epsilon}.$$

\item For $r\not=q+1$ the above results are also true for linear
symmetric reducible testers. For $r=q+1$ result \ref{HPre} is also
true for linear symmetric testers.

\end{enumerate}
\end{lemma}
\begin{proof}
By Lemma \ref{ctb1} and \ref{FFxFFal},
$$\nu_{\FF_q}^\cP(d,\FF_{q^t},\bepsilon)\le
\nu_{\FF_q}^\cP(d,\FF_{q}[X]_{t-1},\bepsilon)\cdot
\nu^\P_{\FF_q[X]_{t-1}}(d,\FF_{q^t},1)
=\nu_{\FF_q}^\cP(d,\FF_{q}[X]_{t-1},\bepsilon).$$ In the same way
$\nu_{\FF_q}^\HP(d,\FF_{q^t},\bepsilon)\le
\nu_{\FF_q}^\HP(d,\FF_{q}[X]_{t-1},\bepsilon).$

We now prove {\it \ref{HPre1}}. For every $f\in \cP(\FF_q,n, d)$
and $(z_1,\ldots,z_n)\in \FF_q[X]_{t-1}^n$ we have
$f(z_1,\ldots,z_n)\in \FF_q[X]_{d(t-1)}$. Let $q\ge d(t-1)+1$.
Choose $F\subseteq \FF_q$ of size $r$, where $q\ge r\ge d(t-1)+1$.
Define for every $\beta\in F$ the map $l_\beta:\FF_q[X]_{t-1}\to
\FF_q$ where $l_\beta(z)=z(\beta)$. If $f(z_1,\ldots,z_n)\not=0$
then since $f(z_1,\ldots,z_n)\in \FF_q[X]_{d(t-1)}$ we have
$l_\beta(f(z_1,\ldots,z_n))=f(l_\beta(z_1),\ldots,l_\beta(z_n))=0$
for at most $d(t-1)$ elements $\beta\in F$. This gives a linear
symmetric $(\cP(\FF_q,n,
d),\FF_q[X]_{t-1},\FF_q)$-$\bepsilon$-tester of size~$r$.
Therefore, for $q\ge d(t-1)+1$,
$$
\nu_{\FF_q}^\cP(d,\FF_{q^t},\bepsilon)\le
\nu_{\FF_q}^\P(d,\FF_q[X]_{t-1},\bepsilon)\le
r=\frac{d(t-1)}{\epsilon}.$$ Notice that the tester is also
reducible since $\l_\beta(1)=1$.

We now prove {\it \ref{HPre}.} For $r$ such that $q\ge r\ge
d(t-1)+1$ the proof is as above. It remains to prove the statement
for $r=q+1$. Consider $f\in \HP(\FF_q,n,d)$ and
$(z_1,\ldots,z_n)\in \FF_q[X]_{t-1}^n$.  Let
$F=\FF_q\cup\{\infty\}$ and define for $z\in\FF_q[X]_{t-1}$,
$l_\beta(z)=z(\beta)$ if $\beta\in \FF_q$ and $\l_\infty(z)$ to be
the coefficient of $X^{t-1}$ in~$z$. Let $L=\{\l_\beta\ |\
\beta\in \FF_q\cup\{\infty\}\}$. It is easy to see that the
coefficient of $X^{d(t-1)}$ in $f(z_1,\ldots,z_n)$ is
$f(l_\infty(z_1),\ldots,l_\infty(z_n))$.

Now suppose $f(z_1,\ldots,z_n)\not=0$. We have two cases: If
$f(l_\infty(z_1),\ldots,l_\infty(z_n))\not=0$ then since
$f(z_1,\ldots,z_n)$ is of degree $d(t-1)$ it can have at most
$d(t-1)$ roots in $\FF_q$. Otherwise, $f(l_\infty(z_1),\ldots,$
$l_\infty(z_n))=0$. Then $f(z_1,\ldots,z_n)$ is of degree at most
$d(t-1)-1$ and can have at most $d(t-1)-1$ roots in $\FF_q$. In
both cases we have that for at most $d(t-1)$ elements $l\in L$,
$f(l(z_1),$ $\ldots,l(z_n))=0$. This gives a linear symmetric
$(\HP(\FF_q,n, d),\FF_q[X]_{t-1},\FF_q)$-$\bepsilon$-tester of
complexity~$r$ which implies the result. Notice that the tester is
not reducible because $\l_\infty(1)=0\not=1$.
\end{proof}

As a consequence of Lemma~\ref{qdt} we get
\begin{corollary} \labell{Cqdt} We have
\begin{enumerate}
\item \label{CHPre1} If $q\ge d(t-1)+1$ then for any $\epsilon<1$
such that $$\epsilon\ge \frac{d(t-1)}{q}$$ we have
$$\nu_{\FF_q}^\cP(d,\FF_{q^t},\bepsilon)\le
\nu_{\FF_q}^\P(d,\FF_q[X]_{t-1},\bepsilon)\le
\left\lceil\frac{d(t-1)}{\epsilon}\right\rceil.$$

\item \label{CHPre} If $q\ge d(t-1)$ then for any any $\epsilon<1$
such that $$\epsilon\ge \frac{d(t-1)}{q+1}$$ we have
$$\nu_{\FF_q}^\HP(d,\FF_{q^t},\bepsilon)\le
\nu_{\FF_q}^\HP(d,\FF_q[X]_{t-1},\bepsilon)\le
\left\lceil\frac{d(t-1)}{\epsilon}\right\rceil.$$

\item \label{CHPre3} Testers of the above densities and sizes can be constructed in linear time $\tilde O(dt/\epsilon)$
and any entry of any of the above maps can be constructed and computed in time $\tilde O(t)$.

\end{enumerate}
\end{corollary}
\begin{proof} We prove {\it \ref{CHPre1}}. The proof of {\it \ref{CHPre}}
is similar.

Let $1>\epsilon>d(t-1)/(d(t-1)+1)$. By Lemma~7 and Theorem~29
in \cite{B1} we have $\nu_{\FF_q}^\cP(d,\FF_{q^t})=d(t-1)+1$ and
by Lemma~\ref{Triv} and (\ref{Tester}) we have
$$\nu_{\FF_q}^\cP(d,\FF_{q^t},\bepsilon)\le
\nu_{\FF_q}^\cP(d,\FF_{q^t},1/(d(t-1)+1))=
\nu_{\FF_q}^\cP(d,\FF_{q^t})=d(t-1)+1=\left\lceil\frac{d(t-1)}{\epsilon}\right\rceil.$$

Now let $d(t-1)/(d(t-1)+1)\ge \epsilon\ge d(t-1)/q$ and let
$r=\lceil d(t-1)/\epsilon\rceil$. Let $\epsilon_1=d(t-1)/r$. Since
$d(t-1)+1\le r=\lceil d(t-1)/\epsilon\rceil\le q$ and
$\epsilon_1\le \epsilon$, by Lemma~\ref{Triv} and Lemma~\ref{qdt}
we have
$$\nu_{\FF_q}^\cP(d,\FF_{q^t},\bepsilon)
\le \nu_{\FF_q}^\cP(d,\FF_{q^t},\bepsilon_1)
\le\left\lceil\frac{d(t-1)}{\epsilon}\right\rceil.$$

Notice that all the constructions in Lemma~\ref{Triv} and Lemma~\ref{qdt} runs
in linear time in $td/\epsilon$.  Computing one entry in a map requires
substituting an element of $\FF_q$ in a polynomial of degree $t$.
This takes time $\tilde O(t)$. This implies~{\it \ref{CHPre3}}.
\end{proof}

The next result shows how to reduce $\bepsilon$-testers for degree
$d$ polynomials in $\FF_{q^t}$ to $\bepsilon$-testers for degree
$d$ polynomials in $\FF_{q^{k}}$ where
$k=O(\log((d/\epsilon)t)/\log q)$. Notice that when
$k=O(\log((d/\epsilon)t)/\log q)$ then
$|\FF_{q^k}|=poly(dt/\epsilon)$. This reduction will be used to
construct $\bepsilon$-testers with almost (within
$poly(d/\epsilon)$) optimal size in polynomial time.

For any positive integer $k$, let $N_q(k)$ denotes the number of
monic irreducible polynomials of degree $k$ over $\FF_q$. It is
known that \begin{eqnarray}\labell{NqkE}
k N_q(k)=\sum_{r|k}\mu\left(\frac{k}{r}\right) q^r
\end{eqnarray}
where $\mu$ is the Moebius function
$$\mu(n)=\left\{
\begin{array}{ll}
1& n=1\\
(-1)^t& n \mbox{\ is the product of $t$ distinct primes}\\
0& \mbox{otherwise}
\end{array}
\right.$$ and
\begin{eqnarray}\labell{Nqk} q^{k-1}<kN_q(k)\le
q^k.\end{eqnarray} See for example \cite{LN}.

We remind the reader that $\nu_{\FF_{q^k}}^{\cP}((d,\FF_q),\FF_{q}[X]_{t-1},\bepsilon_1)$ is $\nu_{\FF_{q^k}}^{\circ}(\P(\FF_q,n,d),\FF_{q}[X]_{t-1},\bepsilon_1)$ which is different
than $\nu_{\FF_{q^k}}^{\cP}(d,\FF_{q}[X]_{t-1},\bepsilon_1)=\nu_{\FF_{q^k}}^{\circ}(\P(\FF_{q^k},n,d),\FF_{q}[X]_{t-1},\bepsilon_1)$.
This notation is used when the ground field is not evident from the context.

We now prove the following

\begin{lemma}\labell{prerec}  We have
\begin{enumerate}
\item \label{cz1} For any finite field $\FF_q$,
any $0<\epsilon_1,\epsilon_2<1$ and integers $k$ and $t$ such that
$$kN_q(k)\ge \frac{dt-d+1}{\epsilon_1},$$ we have
$$\nu^\cP_{\FF_q}(d,\FF_{q^t},\bepsilon_1\bepsilon_2)\le
\left\lceil\frac{dt-d+1}{\epsilon_1\cdot k}
\right\rceil\cdot\nu^{\cP}_{\FF_q} \left(d,\FF_{
q^{k}},\bepsilon_2\right).$$

\item\label{cz2} Given a $(\P(\FF_q,n,d), \FF_{q^k},\FF_q)$-$\bepsilon_2$-tester of
size $s$, one can construct a $(\P(\FF_q,n,d),
\FF_{q^t},\FF_q)$-$\bepsilon_1\bepsilon_2$-tester of size
$$S:=\left\lceil\frac{dt-d+1}{\epsilon_1\cdot k}
\right\rceil\cdot  s $$
in time $$\tilde O\left( S+ k^3p^{1/2}+k^4\right).$$

\item\label{cz3} If constructing and computing any entry of any map in the $(\P(\FF_q,n,d), \FF_{q^k},\FF_q)$-$\bepsilon_2$-tester
takes time $T$ then constructing and computing any entry of any map in the $(\P(\FF,n,d),
\FF_{q^t},\FF_q)$-$\bepsilon_1\bepsilon_2$-tester takes time
$$\tilde O(T+t+k^3p^{1/2}+k^4).$$

\item\label{cz4} If the $(\P(\FF_q,n,d), \FF_{q^k},\FF_q)$-$\bepsilon_2$-tester is componentwise (respectively, linear,
reducible and symmetric tester) then the $(\P(\FF,n,d),
\FF_{q^t},\FF_q)$-$\bepsilon_1\bepsilon_2$-tester is componentwise (respectively, linear,
reducible and symmetric tester).
\end{enumerate}
\end{lemma}
\begin{proof}  By Lemma \ref{qdt} and Lemma
\ref{ctb1} we have
$$\nu_{\FF_q}^{\cP}(d,\FF_{q^t},\bepsilon_1\bepsilon_2)\le
\nu_{\FF_q}^{\cP}(d,\FF_{q}[X]_{t-1},\bepsilon_1\bepsilon_2)\le
\nu_{\FF_{q^k}}^{\cP}((d,\FF_q),\FF_{q}[X]_{t-1},\bepsilon_1)
\cdot\nu_{\FF_q}^{\cP}(d,\FF_{q^k},\bepsilon_2).$$ We now prove
$$\nu_{\FF_{q^k}}^{\cP}((d,\FF_q),\FF_{q}[X]_{t-1},\bepsilon_1)\le
\left\lceil\frac{dt-d+1}{\epsilon_1\cdot k}\right\rceil.$$

Let ${\cal R}$ be the set of all monic irreducible polynomials of
degree $k$. Since
$$deg \left(\prod_{g\in {
\cR}}g\right)= kN_q(k)\ge \frac{dt-d+1}{\epsilon_1},$$ we
can choose ${\cal R}'\subseteq {\cal R}$ such that
$$\frac{dt-d+1}{\epsilon_1}\le deg\ \left(\prod_{g\in { \cR'}}g\right)<
\frac{dt-d+1}{\epsilon_1}+k.$$

Let $\cR''$ be any subset of $\cR'$ where
$${dt-d+1}\le deg\ \left(\prod_{g\in { \cR''}}g\right)<
{dt-d+1}+k.$$

Let $f\in \cP(\FF_q,n,d)$, $z_1,\ldots,z_n\in \FF_q[X]_{t-1}$ and
$F(X):=f(z_1,\ldots,z_n)\in\FF_q[X]_{dt-d}$. Now $F\equiv 0$ if and
only if $F\modd (\prod_{g\in { \cR''}}g)\equiv 0$ if and only if
$F\modd g\equiv 0$ for all $g\in \cR''$. It is known that $F\modd
g\equiv 0$ if and only if
$F(\beta)=f(z_1(\beta),\ldots,z_n(\beta))=0$ for one root
$\beta\in \FF_{q^k}$ of $g$. See for example Theorem~3.33 (ii) in
\cite{LN}.

Define for every $g\in \cR'$ a map $l_\beta:\FF_q[X]\to \FF_{q^k}$
where $\beta\in\FF_{q^k}$ is a root for $g$ and
$l_\beta(z)=z(\beta)$. Let $L$ be the set of all such maps.
Then $|L|=|\cR'|$. We have shown that if
$f(z_1,\ldots,z_n)\not=0$ then $f(l(z_1),\ldots,l(z_n))=0$ for at
most $|\cR''|-1$ maps $l$ in $L$.

Therefore,
$$\nu_{\FF_{q^k}}^{\cP}\left((d,\FF_q),\FF_{q}[X]_{t-1},
1-\frac{|\cR''|-1}{|\cR'|}\right)\le |{\cal R}'|\le
\left\lceil\frac{dt-d+1}{\epsilon_1\cdot k}\right\rceil.$$ Since
$$\frac{|\cR''|-1}{|\cR'|}\le  \frac{\frac{dt-d+1+k}{k}-1}{\frac{dt-d+1}{\epsilon_1k}}
= \epsilon_1,$$ the result follows from Lemma~\ref{Triv}. This implies {\it \ref{cz1}}.

We now describe the construction algorithm and give the time complexity.
The input of the algorithm is some
representation $\FF_{q^t}\simeq \FF_q[\alpha]/(f_1(\alpha))$ for some irreducible polynomial $f_1(x)\in \FF_q[x]$ of degree $t$ and a $(\P(\FF_q,n,d), \FF_{q^k},\FF_q)$-$\bepsilon_2$-tester of
size $s$. Also the field $\FF_{q^k}$ has some representation $\FF_{q^k}\simeq \FF_q[\beta]/(f_2(\beta))$ for some irreducible polynomial $f_2(x)\in \FF_q[x]$ of degree $k$.
The algorithm first define a map $\FF_{q^t}$ to $\FF_q[X]_{t-1}$ that replaces $\alpha$ with $X$. The algorithm then constructs ${\cal R}'$ which is a set of $O(dt/(k\epsilon_1))$ irreducible polynomials of degree $k$ and finds one root in $\FF_{q^k}$ for each polynomial. By Lemma~\ref{ManI}, this takes time $\tilde O((dt/\epsilon_1)+k^3p^{1/2}+k^4\log^2q)$. Then it constructs the maps $l_\beta$ for each root $\beta$. This takes linear time $O(dt/(\epsilon_1 k))$. Then it uses Lemma~\ref{ctb1} which takes linear time in the total size $O(sdt/(\epsilon_1 k))$. Since $s\ge k$ we have $\tilde O(dt/\epsilon_1)=\tilde O(S)$. This gives the time complexity. This implies {\it \ref{cz2}}

For accessing one map $l_\beta$ in the tester we need to construct the
$i$th irreducible polynomial. By Lemma~\ref{ManIle}, this can be done
in time $O(k^3p^{1/2}+k^4)$. Computations in the fields $\FF_{q^t}$ and $\FF_{q^k}$
and the map from $\FF_{q^t}$ to $\FF_q[X]_{t-1}$ take time $\tilde O(t)$.
This gives the complexity $\tilde O(t+k^3p^{1/2}+k^4)$. This implies {\it \ref{cz3}}

By Lemma~\ref{CLRS}, {\it \ref{cz4}} is immediate from the construction.
\end{proof}

We note that a slightly better bound can be obtained if $\cR$ is
the set of all the monic irreducible polynomials of degree at most
$k$.
When $k$ divides $t$, a better bound is proved in the following.
We will not use this result in this paper so we will not 
bother the reader with an almost linear time or local explicit construction
and just give the proof for the poly-time construction

\begin{lemma}\label{prerec2}  For any finite field $\FF_q$,
any $0<\epsilon_1,\epsilon_2<1$ and integers $k$ and $t$ such that
$k|t$ and
$$kq^{k}> \frac{d(t-k)}{\epsilon_1},$$ we have
$$\nu^\cP_{\FF_q}(d,\FF_{q^t},\bepsilon_1\bepsilon_2)\le
\left\lceil\frac{d(t-k)}{\epsilon_1\cdot k}
\right\rceil\cdot\nu^{\cP}_{\FF_q} \left(d,\FF_{
q^{k}},\bepsilon_2\right).$$

Given a $(\P(\FF,n,d), \FF_{q^k},\FF_q)$-$\bepsilon_2$-tester of
size $s$, one can construct a $(\P(\FF,n,d),
\FF_{q^t},\FF_q)$-$\bepsilon_1\bepsilon_2$-tester of size $s\cdot\lceil{d(t-k)}/({\epsilon_1\cdot k})
\rceil$
in time $(sd/\epsilon_1)\cdot poly(t,k,p,\log
q)$.
\end{lemma}
\begin{proof} By Corollary \ref{ctb11} we have
$$\nu_{\FF_q}^\cP(d,\FF_{q^t},\bepsilon_1\bepsilon_2)\le
\nu_{\FF_{q^k}}^\cP(d,\FF_{(q^k)^{t/k}},\bepsilon_1)
\cdot\nu_{\FF_q}^\cP(d,t,\FF_{q^k},\bepsilon_2) .$$

Let $r=\lceil d(t-k)/(\epsilon_1k)\rceil$. By Lemma \ref{qdt},
since $q^k\ge r\ge d(t/k-1)+1$ for $\epsilon'=d(t/k-1)/r$, we have
$$\nu_{\FF_{q^k}}^\cP(d,\FF_{(q^k)^{t/k}},\overline{\epsilon_1})
\le \nu_{\FF_{q^k}}^\cP(d,\FF_{(q^k)^{t/k}},\overline{\epsilon'})
\le \frac{d(t/k-1)}{\epsilon'}\le r.$$

Given a $(\P(\FF,n,d), \FF_{q^k},\FF_q)$-$\bepsilon_2$-tester
where $\FF_{q^k}=\FF_q[u]/(g(u))$ where $g(u)$ is irreducible
polynomial in $\FF_q$ of degree $k$. By Lemma~\ref{FF1} the field
$\FF_{(q^k)^{t/k}}$ can be constructed in time $poly(p,t/k,k\log
q)$. By Lemma~\ref{qdt} the $(\P(\FF,n,d),
\FF_{(q^k)^{t/k}},\FF_{q^k})$-$\bepsilon_1$-tester can be
constructed in time $O(dt/\epsilon_1)$.
\end{proof}

\section{Lower Bounds}\label{LBsection}\labell{s12}

In this section we give some lower bounds for the complexity of
$\bepsilon$-tester. Then we give some lower bound for the density $\bepsilon$
for which an $\bepsilon$-tester exists.

\subsection{Lower Bound for the Size}\labell{s13}
We first prove

\begin{theorem}\labell{TRIVLB} Let $\FF$ be a field and $\cA$ and $\cB$
be two $\FF$-algebras. Let $0\le \epsilon< 1$ and ${\cal
M}\subseteq \FF[x_1,x_2,\ldots,x_n]$ be a class of multivariate
polynomial. Let $S\subseteq \cA$ and $R\subseteq \cB$ be linear
subspaces. Then
$$\nu_R^\circ(\cM,S,\bepsilon)\ge
\frac{\nu_R^\circ(\cM,S)-1}{\epsilon}.$$

In particular,
$$\nu^\cP_{\FF_q}(d,\FF_{q^t},\bepsilon)\ge \nu^\HP_{\FF_q}(d,\FF_{q^t},\bepsilon)
\ge \frac{d(t-1)}{\epsilon},$$ and for $q=o(d)$
$$\nu_{\FF_q}(d,\FF_{q^t},\bepsilon)\ge
\frac{\left(1+\frac{1}{q-1}-\frac{1}{(q-1)q^{t-1}}\right)^{d-1}}{\epsilon}\cdot
t=\frac{2^{\Omega\left(\frac{1}{q}\right)d}}{\epsilon}\cdot t.$$
\end{theorem}
\begin{proof} If $L$ is an optimal $(\cM,S,R)$-$\bepsilon$-tester then any
$L'\subseteq L$ of size $|L'|=\lfloor \epsilon |L|\rfloor+1$ is a
$(\cM,S,R)$-tester. Therefore,
$$\lfloor \epsilon |L|\rfloor+1\ge \nu_R^\circ(\cM,S).$$
Since $\lfloor \epsilon |L|\rfloor+1\le \epsilon |L|+1$ the result
follows.

The other results follows from Theorem~27 and~29 in \cite{B1}.
\end{proof}
By Theorem~\ref{TRIVLB} and Corollary~\ref{Cqdt} we
have
\begin{corollary} \labell{CqdtT} We have
\begin{enumerate}
\item \label{CHPre1} If $q\ge d(t-1)+1$ then for any $\epsilon<1$
such that $$\epsilon\ge \frac{d(t-1)}{q}$$ we have
$$\nu_{\FF_q}^\cP(d,\FF_{q^t},\bepsilon)=
\left\lceil\frac{d(t-1)}{\epsilon}\right\rceil.$$

\item \label{CHPre} If $q\ge d(t-1)$ then for any $\epsilon<1$
such that $$\epsilon\ge \frac{d(t-1)}{q+1}$$ we have
$$\nu_{\FF_q}^\HP(d,\FF_{q^t},\bepsilon)=
\left\lceil\frac{d(t-1)}{\epsilon}\right\rceil.$$

\end{enumerate}
\end{corollary}

We say that $C$ is a {\it hitting set} over $\FF_q$
for $\cM$ of density $1-\epsilon$ if $C\subseteq \FF_q^n$
and for every $f\in \cM$
there are at least $(1-\epsilon)|C|$ elements $c$ in $C$ such that $f(c)\not=0$.

For $(\HLF(\FF_q,n,d),\FF_{q^t},\FF_q)$-$\bepsilon$-tester we give
the following better bound

\begin{theorem} \labell{LB22}
For any $q$, $d$ and $t$ we have
$$\nu_{\FF_q}(d,\FF_{q^t},\bepsilon)\ge
\frac{t-1}{\left(1-\frac{1}{q}+\frac{q-1}{q(q^t-1)}\right)^{d-1}-(1-\epsilon)}.$$

\ignore{\item \label{LB222} For any $q\ge d+1$ and $t$,
$$\nu_{\FF_q}^\cP(d,\FF_{q^t},\bepsilon)\ge
\frac{t-1}{\epsilon-\frac{d-1}{q}}.$$

\item \label{LB223} For any $q\ge d$ and $t$,
$$\nu_{\FF_q}^\HP(d,\FF_{q^t},\bepsilon)\ge
\frac{t-1}{\epsilon-\frac{d-1}{q+1}}.$$
\end{enumerate}}
\end{theorem}
\begin{proof} Consider the class of functions
$$\cM=\left\{\left.\left(\prod_{i=1}^{d-1}\sum_{m=1}^{t}\lambda_{i,m}y_{i,m}
\right) (y_{d,k}-y_{d,j})\  \right| \ \bflambda_i\in P^t(\FF_q)
\mbox{\ for all $i=1,\ldots,d-1$},\ 1\le k<j\le q^t\right\},$$
where $P^t(\FF_q)$ is the $t$-dimensional projective space over
$\FF_q$. For $\blambda=(\blambda_{1},\blambda_{2}\ldots,
\blambda_{d-1})\in P^t(\FF_q)^{d-1}$ we denote
$f_\blambda=\prod_{i=1}^{d-1}(\sum_{m=1}^{t}\lambda_{i,m}y_{i,m})$.
Let $\cM'=\{(y_{d,k}-y_{d,j})\ | \ 1\le k<j\le q^t \}.$

Obviously, $\cM\subseteq \HLF(\FF_q,n,d)$ where $n=q^t$. Let
$L=\{\bell_1,\ldots,\bell_\nu\}$ be a
$(\HLF(\FF_q,n,d),\FF_{q^t},\FF_q)$-$\bepsilon$-tester of minimum
size. Then $L$ is an $(\cM,\FF_{q^t},\FF_q)$-$\bepsilon$-tester.
Let $\alpha$ be a primitive element of $\FF_{q^t}$ and consider
the assignment $\bfz=(\bfz_1,\ldots,\bfz_d)\in (\FF_{q^t}^n)^d$
where $\bfz_i=(\alpha^0,\alpha^1,{\ldots},\alpha^{t-1},0,\ldots,0)
\in\FF_{q^t}^{n}$ for all $i=1,2,\ldots,d-1$ and
$\bfz_d=(0,\alpha^0,\alpha^1,\ldots,\alpha^{q^t-2}) \in
\FF_{q^t}^{n}$. Let $\bfc_i=\bfell_i(\bfz)\in(\FF_{q}^{n})^d$ for
$i=1,\ldots,\nu$ and $C=\{\bfc_i\ |\ i=1,2,\ldots,\nu\}$. Since
$f(\bfz)\not=0$ for all $f\in \cM$ and $L$ is a
$(\cM,\FF_{q^t},\FF_q)$-$\bepsilon$-tester, for every $f\in \cM$
there are at least $(1-\epsilon)|C|$ elements $\bfc\in C$ such
that $f(\bfc)\not=0$. Therefore, $C$ is a hitting set over $\FF_q$
for $\cM$ of density $1-\epsilon$.

Notice that if for some $\bfc\in C$ and some $i=1,2,\ldots,d-1$ we have
$(c_{i,1},c_{i,2},\ldots,c_{i,t})=0$
then $f(\bfc)=0$ for all $f\in \cM$ and then $C\backslash
\{\bfc\}$ is a hitting set over $\FF_q$ for $\cM$ of density at
least $\bepsilon$. Therefore we may assume w.l.o.g that
$(c_{i,1},c_{i,2},\ldots,c_{i,t})\not=0$ for all $\bfc\in C$ and
$i=1,2,\ldots,d-1$.

Now for every $\blambda\in P^t(\FF_q)^{d-1}$ consider the set
$C_\blambda=\{ \bfc\in C\ |\ f_\blambda(\bfc)\not=0\}.$ For every
$C_\blambda=\{\bfc^{(1)},\ldots,\bfc^{(r)}\}$ consider the set
$$D_\blambda=\left\{\left.(c^{(1)}_{d,i},\ldots,c^{(r)}_{d,i})\ \right|\
i=1,\ldots,q^t\right\}.$$

Since for any $1\le k<j\le
q^t$, $C$ is a hitting set for $f_{\bflambda}(y_{d,k}-y_{d,j})$
of density $1-\epsilon$, for every $1\le k<j\le
q^t$ there are at least $(1-\epsilon)|C|$ elements $\bfc\in
C_\blambda$ such that $c_{d,k}\not=c_{d,j}$. Thus, $D_\blambda$ is
a code of Hamming distance $(1-\epsilon)|C|$. By the Singleton
bound, \cite{Ro06}, we have
$$t=\frac{\log |D_\blambda|}{\log q}\le
|C_\blambda|-(1-\epsilon)|C|+1$$ and therefore $|C_\blambda|\ge
(t-1)+(1-\epsilon)|C|$. Now it is easy to see that since
$(c_{i,1},c_{i,2},\ldots,c_{i,t})\not=0$ for all $\bfc\in C$ and
$i=1,2,\ldots,d-1$, every $\bfc\in C$ appears in exactly
$$\left(\frac{q^t-1}{q-1}-\frac{q^{t-1}-1}{q-1}\right)^{d-1}=
\left( \frac{q^t-q^{t-1}}{q-1}\right)^{d-1}$$ sets of
$\{C_\blambda\ |\ \bflambda\in P^t(\FF_q)^{d-1}\}$. Therefore
\begin{eqnarray*}
|C|\ge \frac{\sum_\blambda |C_\blambda|}{\left(
\frac{q^t-q^{t-1}}{q-1}\right)^{d-1}}&\ge&
\frac{|P^t(\FF_q)^{d-1}|\cdot ((t-1)+(1-\epsilon)|C|)}{\left(
\frac{q^t-q^{t-1}}{q-1}\right)^{d-1}}\\
&=& \frac{\left(\frac{q^t-1}{q-1}\right)^{d-1}\cdot
((t-1)+(1-\epsilon)|C|)}{\left(
\frac{q^t-q^{t-1}}{q-1}\right)^{d-1}}\\
&=&
\left(\frac{q^t-1}{q^t-q^{t-1}}\right)^{d-1}((t-1)+(1-\epsilon)|C|).
\end{eqnarray*}

Therefore,
$$\nu_{\FF_q}(d,\FF_{q^t},\bepsilon)=\nu=|C|\ge
\frac{t-1}{\left(1-\frac{1}{q}+\frac{q-1}{q(q^t-1)}\right)^{d-1}
-(1-\epsilon)}.$$ This proves the result.

\ignore{To prove {\it \ref{LB222}} we take
$$\cM=\left\{\left.\left(\prod_{i=1}^{d-1}(x_3-\lambda_i)
\right) (x_{k}-x_{j})\  \right| \ \lambda_1,\ldots,\lambda_{d-1}
\in \FF_q \mbox{\ are distinct},\ 1\le k<j\le q^t\right\}$$ and
$\bfz=(0,\alpha^0,\alpha^1,\ldots,\alpha^{q^t-2})$.

To prove {\it \ref{LB223}} we take
$$\cM=\left\{\left.\left(\prod_{i=1}^{d-1}
(\lambda_{i,1}x_3-\lambda_{i,2}x_2) \right) (x_{k}-x_{j})\ \right|
\  \bflambda_{1},\ldots,\bflambda_{d-1} \in P^2(\FF_q) \mbox{\ are
distinct},\ 1\le k<j\le q^t\right\}$$ and
$\bfz=(0,\alpha^0,\alpha^1,\ldots,\alpha^{q^t-2})$.}
\end{proof}

Note that for $Y=1-1/q+(q-1)/(q(q^t-1))$ and $\bepsilon=\delta
Y^d$ for some $\delta<1$ the bounds in Theorem~\ref{TRIVLB} and
Theorem~\ref{LB22} are
$$\frac{t}{(1-\delta Y^d)Y^{d-1}},\ \ \  \frac{t-1}{(1-\delta
Y)Y^{d-1}}$$ respectively. Therefore the later bound is slightly
better than the former.

\subsection{Lower Bound on the Density}\labell{s14}

How small $\epsilon$ can be? In the following we give a lower
bound for $\epsilon$.
\begin{theorem} \labell{lowerB} We have
\begin{enumerate}
\item \label{ftt1} If there is a
$(\HLF(\FF_q,n,d),\FF_{q^t},\FF_q)$-$\bepsilon$-tester of finite
size then
$$\epsilon\ge 1- \left(1-\frac{1}{q}+
\frac{q-1}{q(q^t-1)}\right)^d.$$ For $d=o(q)$
$$\epsilon\ge\frac{d}{q}-\Theta\left(\frac{d^2}{q^2}\right).$$
Therefore if $\bepsilon> (1-1/q+(q-1)/(q(q^t-1)))^d$ then
$\nu_{\FF_q}(d,\FF_{q^t},\bepsilon)=\infty.$

\item \label{ftt2} If there is a
$(\P(\FF_q,n,d),\FF_{q^t},\FF_q)$-$\bepsilon$-tester of finite
size then  $$\epsilon\ge \frac{d}{q}.$$ Therefore if
$\bepsilon>1-d/q$ then for any $t$ we have
$\nu^\P_{\FF_q}(d,\FF_{q^t},\bepsilon)=\infty.$

\item \label{ftt22} For $d\ge q$ and any $t$ and $\epsilon$ we
have $\nu^\P_{\FF_q}(d,\FF_{q^t},\bepsilon)=
\nu^\P_{\FF_q}(d,\FF_{q^t})=\infty.$

\item \label{ftt3} If there is a
$(\HP(\FF_q,n,d),\FF_{q^t},\FF_q)$-$\bepsilon$-tester of finite
size then $$\epsilon\ge \frac{d}{q+1}.$$ Therefore, if $\bepsilon>
1-{d}/{(q+1)}$ then for any $t$ we have
$\nu^\HP_{\FF_q}(d,\FF_{q^t},\bepsilon)=\infty.$

\item \label{ftt32} For $d\ge q+1$ and any $t$ and $\epsilon$ we
have $\nu^\HP_{\FF_q}(d,\FF_{q^t},\bepsilon)=
\nu^\HP_{\FF_q}(d,\FF_{q^t})=\infty.$
\end{enumerate}
\end{theorem}
\begin{proof}
We first prove {\it \ref{ftt1}}. Consider the class of functions
$$\cM=\left\{\left.\prod_{i=1}^{d}\sum_{m=1}^{t}\lambda_{i,m}y_{i,m}
 \  \right| \ \bflambda_i\in P^t(\FF_q) \mbox{\ for all
$i=1,\ldots,d$}\right\},$$ where $P^t(\FF_q)$ is the
$t$-dimensional projective space over $\FF_q$. For
$\blambda=(\blambda_{1},\blambda_{2}\ldots, \blambda_{d})\in
P^t(\FF_q)^{d}$ we denote
$f_\blambda=\prod_{i=1}^{d}(\sum_{m=1}^{t}\lambda_{i,m}y_{i,m})$.
Obviously, $\cM\subseteq \HLF(\FF_q,n,d)$. Let
$L=\{\bell_1,\ldots,\bell_\nu\}$ be a
$(\HLF(\FF_q,n,d),\FF_{q^t},\FF_q)$-$\bepsilon$-tester of minimum
size. Then $L$ is an $(\cM,\FF_{q^t},\FF_q)$-$\bepsilon$-tester.
Let $\alpha$ be an element of degree $t$ in $\FF_{q^t}$ and
consider the assignment $\bfz=(\bfz_1,\ldots,\bfz_d)\in
(\FF_{q^t}^n)^d$ where
$\bfz_i=(\alpha^0,\alpha^1,{\ldots},\alpha^{t-1},0,\ldots,0)
\in\FF_{q^t}^{n}$ for all $i=1,2,\ldots,d$. Notice that
$f(\bfz)\not=0$ for all $f\in\cM$. Let
$S=\{\bfl_1(\bfz),\ldots,\bfl_\nu(\bfz)\}\subseteq (\FF_q^n)^d$.
We now show that there is $f\in \cM$ such that $|\{\bfa\in S\ |\
f(\bfa)\not=0\}|\le (1-1/q+(q-1)/(q^{t+1}-q))^d |S|$.

Define a sequence of sets $S=S_0\supseteq
S_1\supseteq\ldots\supseteq S_d$ recursively as follows: For the set $S_i$,
$i>0$, consider the functions $\sum_{j=1}^{t} \lambda_{i,j}
y_{i,j}$ where $\blambda_i\in P^t(\FF_q)$. There is
$\blambda_i'\in P^t(\FF_q)$ such that $f_i=\sum_{j=1}^{t}
\lambda_{i,j}' y_{i,j}$ is zero on at least
$|S_{i-1}|(q^{t-1}-1)/(q^t-1)$ elements of $S_{i-1}$. Define
$S_i=S_{i-1}\backslash \{\bfa\in S_{i-1}\ |\ f_i(\bfa)=0\}$. Then
$$|S_i|\le \frac{q^t-q^{t-1}}{q^t-1}|S_{i-1}|
=\left(1-\frac{1}{q}+\frac{q-1}{q(q^t-1)}\right)|S_{i-1}|.$$

Then $f=f_1f_2\cdots f_d\in \cM$ is not zero only on the elements
of $S_d$ and $|S_d|\le (1-1/q+(q-1)/(q^{t+1}-q))^d|S|$. This
implies the result.

To prove {\it \ref{ftt2}} we take the class $$\cM=\{f\ |\
f=f_1f_2\cdots f_d,\ f_i= x_1-\beta_i,\ \beta_i\in
\FF_q\}\subset\P(\FF_q,n,d).$$ Let
$L=\{\bell_1,\ldots,\bell_\nu\}$ be a
$(\cP(\FF_q,n,d),\FF_{q^t},\FF_q)$-$\bepsilon$-tester of minimum
size. Let $S=\{\bfl_1(\bfz),\ldots,\bfl_\nu(\bfz)\}\subseteq
\FF_q^n$ where $\bfz=(\alpha,0,0,\ldots,0)\in \FF_{q^t}^n$ and $\alpha\in \FF_{q^t}\backslash \FF_q$. Then
$f(\bfz)\not=0$ for all $f\in \cM$ and there are
$\beta_1,\beta_2,\ldots,\beta_d\in \FF_q$ such that $f$ is not
zero on at most $1-d/q$ fraction of the elements of $S$.

{\it \ref{ftt22}} follows from Lemma~28 in \cite{B1}.

To prove {\it \ref{ftt3}} we take a function of the form
$f=f_1f_2\cdots f_d\in\HP(\FF_q,n,d)$ where $f_i=\gamma_i
x_1-\beta_i x_2$,
$(\gamma_i,\beta_i)\in\FF_q^2\backslash\{(0,0)\}$ and
$\bfz=(1,\alpha,0,\ldots,0)$ where $\alpha\in \FF_{q^t}\backslash \FF_q$.
It is easy to see that there is such
function that is not zero on at most $1-d/(q+1)$ fraction of the
elements of $S$.

{\it \ref{ftt32}} follows from Lemma~28 in \cite{B1}.
\end{proof}

We note here that for {\it \ref{ftt3}} in Theorem \ref{lowerB} the
following slightly better bound
$$\left(1-\frac{1}{q}\right)^d\left(1+\frac{1}{q^t-1}\right)=1-\frac{d}{q}+
\Theta\left(\frac{d^2}{q^2}\right)$$ can be proved if we take
$f=\prod_{i=1}^d\sum_{j=1}^t \lambda_{i,j}x_j$ where
$(\lambda_{i,j})_{i,j}$ is a matrix of rank $d$.

\section{Constructions of Dense Testers}\labell{s15}

In this section we give some constructions of dense testers.

In Subsection~\ref{sss} we give several constructions for testers for $\P(\FF_q,n,d)$
from $\FF_{q^t}$ to $\FF_q$. By Theorem~\ref{lowerB}, such constructions
exist if $q\ge d+1$ and $\bepsilon \le 1-d/q$. Our constructions give testers
of sizes that are within $poly(d/\epsilon)$ factor
from optimal with any density $\bepsilon\le 1-d/q-d/q^2-o(d/q^2)$.

In Subsection~\ref{SField} we give a construction for tester for $\HLF(\FF_q,n,d)$
from $\FF_{q^t}$ to $\FF_q$ for any $q$. Theorem~\ref{TRIVLB} and  Theorem~\ref{lowerB}
show that the size
of such tester is at least $(1+1/(q-1))^dt$ and its density is at most $\bepsilon\le (1-1/q)^d$.
We give a tester of size $(1+(\log q)/q)^dt$ of density $\bepsilon\le (1-(\log q)/q)^d$.

Section~\ref{s18} shows how to construct such testers in almost linear time and shows
that such constructions are locally explicit.

\subsection{Dense Testers for Large Fields}\labell{s16}
\label{sss}

In this section we use algebraic function fields to
construct an $\bepsilon$-testers for large fields.

We prove
\begin{theorem}\labell{T2} For any $q\ge d+1$,
any $t$, any constant $c$ and any
$\epsilon>d/q+d/q^2+d/q^{2^2}+\cdots+d/q^{2^c}+8d/(q^{2^{c+1}}-1)$,
we have
$$\nu^\cP_{\FF_q}(d,\FF_{q^t},\bepsilon)
\le {poly\left(\frac{d}{\epsilon}\right)\cdot t}.$$

In particular, the bound holds for
$\nu_{\FF_q}(d,\FF_{q^t},\bepsilon)$ and
$\nu^\HP_{\FF_q}(d,\FF_{q^t},\bepsilon)$.
\end{theorem}

The exact $poly(d/\epsilon)$ is given in the following two
theorems. Theorem~\ref{T2n} is for $q>10\cdot d$ and
Theorem~\ref{T3n} is for $d+1\le q\le 10d$. Notice that by
Theorem~\ref{lowerB}, $\epsilon\ge d/q$ and therefore when $d+1\le
q\le 10d$ we have $\epsilon=O(1)$ and $poly(d/\epsilon)=poly(d)$. This
is why $\epsilon$ does not appear in the size of the testers in Theorem~\ref{T3n}.

\begin{theorem}\labell{T2n} For any $q\ge 10\cdot d$
and any constant $c>1$ we have the following results
\begin{center}
\begin{tabular}{|c|c|c|l|l|}
$q$& $t$& $\epsilon$&$\nu_{\FF_q}^\P(d,\FF_{q^t},\epsilon)=O(\cdot)$ & Result\\
\hline P.S.& I.S.& $\epsilon\ge \frac{d}{\sqrt{q}-1}$& $\frac{d}{\epsilon}\cdot t$ &Lemma~\ref{LT01} \\
\hline $-$& I.S. & $\epsilon\ge 2\frac{d}{{q}}$
& $\left(\frac{d}{\epsilon}\right)^2\cdot t$  & Lemma~\ref{LT02}\\
\hline $-$& $-$ & $\epsilon\ge 8\frac{d}{{q}}$ &
$\left(\frac{d}{\epsilon}\right)^3\cdot t$  & Lemma~\ref{LT03}\\
\hline $-$& $-$ & $\epsilon\ge c\frac{d}{{q}}$
& $\left(\frac{d}{\epsilon}\right)^4\cdot t$  & Lemma~\ref{LT04}\\
\hline $-$& $-$& $\epsilon\ge
\frac{d}{{q}}+o\left(\frac{d}{q}\right)$
& $\left(\frac{d}{\epsilon}\right)^{4+o(1)}\cdot t$  & Lemma~\ref{LT04}\\
\hline $-$& $-$ & $\epsilon\ge
\frac{d}{{q}}+O\left(\frac{d}{q^2}\right)$&
$\left(\frac{d}{\epsilon}\right)^{9}\cdot t$  &
Lemma~\ref{LT05}\\
\hline $-$& $-$ & $\epsilon\ge
\frac{d}{{q}}+\frac{d}{q^2}+o\left(\frac{d}{q^2}\right)$&
$\left(\frac{d}{\epsilon}\right)^{9+o(1)}\cdot t$  &
Lemma~\ref{LT05}\\
\hline
\end{tabular}
\end{center}
In the table, P.S. stands for ``perfect square'' and I.S. stands
for ``for infinite sequence of integers''.
\end{theorem}

\begin{theorem}\labell{T3n} For any $10\cdot d \ge q=d+\delta$ where $\delta\ge 1$
and any constant $c<1$ we have the following results
\begin{center}
\begin{tabular}{|c|c|l|}
$t$ & $\epsilon$ & $\nu_{\FF_q}^\P(d,\FF_{q^t},\epsilon)=O(\cdot)$ \\
\hline
I.S. & $\epsilon\ge 1-\frac{c\delta}{q}=\frac{d}{q}+(1-c)\frac{\delta}{q}$ & $d^3\cdot t$\\
\hline
$-$ & $\epsilon\ge 1-\frac{c\delta}{q}=\frac{d}{q}+(1-c)\frac{\delta}{q}$ & $d^4\cdot t$\\
\hline
I.S. & $\epsilon\ge 1-\frac{\delta}{q}+o\left(\frac{\delta}{q}\right)
=\frac{d+o(\delta)}{q}$ & $d^{3+o(1)}\cdot t$\\
\hline $-$ & $\epsilon\ge
1-\frac{\delta}{q}+o\left(\frac{\delta}{q}\right)
=\frac{d+o(\delta)}{q}$ & $d^{4+o(1)}\cdot t$\\
\hline
\end{tabular}
\end{center}
\end{theorem}

In Theorem \ref{lowerB} we have shown that for $\epsilon< d/q$ or $q<
d+1$ there is no $\bepsilon$-tester for $\P(\FF_q,n,d)$. This
shows that the bound $q\ge d+1$ in Theorem~\ref{T2} and~\ref{T3n}
is tight and $\epsilon$ is almost tight. In Theorem~\ref{TRIVLB}
we have shown that $\nu_{\FF_q}^\cP(d,\FF_{q^t},\bepsilon)\ge
d(t-1)/\epsilon$. So the size
of our $\bepsilon$-tester is within a $poly(d/\epsilon)$ factor of
the optimal size.

\ignore{In Result \ref{res1} we show that for a perfect square $q$
and $\epsilon\ge  2d/(\sqrt{q}-1)$ there is an infinite sequence
of integers $t$ such that an $\bepsilon$-tester that is optimal
within a factor of $(1+o(1))d$ can be constructed .
Results~\ref{res2}~and~\ref{res3} achieve a factor of
$O(d^2/\epsilon)$ for other ranges of $q$, $\epsilon$ and $t$.

In Section~\ref{PolyTime} we show that for every $d,q,t$ and
$\epsilon$ as in Theorem \ref{T2}, there exists a constant
$\tau(d,q,t,\epsilon)\le 10$ such that a tester of complexity
$(d/\epsilon)^{\tau(d,q,t,\epsilon)}\cdot t$ can be constructed in
polynomial time. Minimizing $\tau(d,q,t,\epsilon)$ requires more
results in algebraic function fields and algebraic geometry and
will be studied in~\cite{B2}.}

For $\nu_{\FF_q}^\HP$ (and therefore for $\nu_{\FF_q}$) slightly better
results can be obtained.

\begin{theorem}\labell{T2HP} For any $q\ge d+1$,
any $t$, any constant integer $c$ and any
$\epsilon>d/(q+1)+d/(q^2+1)+d/(q^{2^2}+1)+\cdots+d/(q^{2^c}+1)+8d/(q^{2^{c+1}}-1)$,
we have
$$\nu^\HP_{\FF_q}(d,\FF_{q^t},\bepsilon)
\le {poly\left(\frac{d}{\epsilon}\right)\cdot t}.$$

In particular, the bound holds for
$\nu_{\FF_q}(d,\FF_{q^t},\bepsilon)$.
\end{theorem}

\begin{theorem}\labell{T3nHP} For any $10\cdot d \ge q=d+\delta-1$ where $\delta\ge 1$
and any constant $c<1$ we have the following results
\begin{center}
\begin{tabular}{|c|c|l|}
$t$ & $\epsilon$ & $\nu_{\FF_q}^\HP(d,\FF_{q^t},\epsilon)=O(\cdot)$ \\
\hline
I.S. & $\epsilon\ge 1-\frac{c\delta}{q+1}=\frac{d}{q+1}+(1-c)\frac{\delta}{q+1}$ & $d^3\cdot t$\\
\hline
$-$ & $\epsilon\ge 1-\frac{c\delta}{q+1}=\frac{d}{q+1}+(1-c)\frac{\delta}{q+1}$ & $d^4\cdot t$\\
\hline
I.S. & $\epsilon\ge 1-\frac{\delta}{q+1}+o\left(\frac{\delta}{q+1}\right)
=\frac{d+o(\delta)}{q+1}$ & $d^{3+o(1)}\cdot t$\\
\hline $-$ & $\epsilon\ge
1-\frac{\delta}{q+1}+o\left(\frac{\delta}{q+1}\right)
=\frac{d+o(\delta)}{q+1}$ & $d^{4+o(1)}\cdot t$\\
\hline
\end{tabular}
\end{center}
\end{theorem}

Now notice that $d/(q+1)=d/q+\theta(d/q^2)$ and therefore the bounds for
$\epsilon$ in
Theorem~\ref{T2n} (except the last row in the table that is included
in Theorem~\ref{T3nHP}) are the same for $\nu_{\FF_q}^\HP$.

We also show
\begin{theorem}\label{RRR}
All the above bounds are true for componentwise, linear reducible and symmetric
$\bepsilon$-testers.
\end{theorem}

We note here that there are many other results that are not
included in the above theorems. For example, for infinite sequence
of integers $t$, any constant $c>1$ and $\epsilon\ge
\epsilon_{min}=cd/q$ we have
$\nu_{\FF_q}^\P(d,\FF_{q^t},\epsilon)\le (d/\epsilon)^3 t$. We
simply avoided results that immediately follows from the above
results and their proof techniques.

For notations used in this section we refer the reader to Sections
$1.1-1.4$ in \cite{Stic08}.

We first prove

\begin{lemma} \label{cvaff1} Let $F/\FF_q$ be a function field,
$P_1,\ldots,P_s$ be distinct places of $F/\FF_q$ of degree $1$ and
$D=P_1+P_2+\cdots+P_s$. Let $G$ be a divisor of $F/\FF_q$ such
that $(\Supp D)\cap (\Supp G)=\O$. Let
$L=\{l_{P_1},\ldots,l_{P_s}\}$ be a set of maps $l_{P_i}:\L(G)\to
\FF_q\cup\{\infty\}$ where $l_{P_i}(x):=x(P_i)$. If $s>d\deg(G)$
then $L$ is a componentwise, linear, reducible and symmetric
$(\cP(\FF_q,n,d),\L(G),\FF_q)$-$(1-d\deg G/s)$-tester of size $s$.
Therefore
$$\nu^\cP_{\FF_q}(d,\L(G),1-d\deg G/s)\le s.$$
\end{lemma}
\begin{proof}  We have shown in Lemma~12.1 in \cite{B1} that any
$L'\subseteq L$ where $|L'|=d\deg(G)+1$ is a symmetric and
reducible $(\cP(\FF_q,n,d),\L(G),\FF_q)$-tester. This implies the
result.
\end{proof}

%
%

In Lemma~13 in \cite{B1} we have proved

\begin{lemma} \label{cvaff2} Let $F/\FF_q$ be a function field.
Let $G$ be a divisor of $F/\FF_q$ and $Q$ a prime divisor of
degree $\deg Q=\ell(G)=t$ such that $v_Q(G)=0$. If $\ell(G-Q)=0$
then we have
\begin{enumerate}
\item \label{fc1} The map
\begin{eqnarray*}
E:{\L}(G)&\rightarrow&F_Q\cong\FF_{q^t}\\
f&\mapsto&f(Q)
\end{eqnarray*} is an isomorphism of linear spaces over $\FF_q$

\item\label{fc2} $L=\{E^{-1}\}$ is a linear, reducible and symmetric
$(\FF_q[\bfx],\FF_{q^t},\L(G))$-$1$-tester where
$\bfx=(x_1,\ldots,x_n)$. Therefore
$$\nu^\circ_{\L(G)}(\FF_q[\bfx],\FF_{q^t},1)=1.$$
\end{enumerate}
\end{lemma}

We now use the above two lemmas to prove

\begin{lemma} \labell{cvaff3} Let $d$ and $t$ be two integers.
Let $F/\FF_q$ be a function field of genus
$g$ that has $N>d(t+g-1)$ places of degree $1$. If $t\ge 3+2\log_q(2g+1)$
then for any $N\ge s>d(t+g-1)$ and $\epsilon=d(t+g-1)/s$ we have
$$\nu^\cP_{\FF_q}(d,\FF_{q^t},\bepsilon)\le \frac{d(t+g-1)}{\epsilon}.$$
\end{lemma}
\begin{proof} First, by Corollary 5.2.10 (c) in \cite{Stic08},
if $2g + 1 \le q^{(t-1)/2}(q^{1/2}-1)$ then there is a prime
divisor of degree~$t$. Since $t \ge 3 + 2 \log_q(2g + 1)$ the
inequality holds and there is at least one prime divisor of degree
$t$. Let $Q$ be such divisor. Let $P_1,\ldots,P_s$, $s>d(t+g-1)$,
be distinct places of $F/\FF_q$ of degree $1$ and
$D=P_1+P_2+\cdots+P_s$. By Lemma~14 in \cite{B1} and Lemma 2.1 and
2.2 in \cite{B99}, there is a divisor $G$ of $F/\FF_q$ such that
$(\Supp D)\cap (\Supp G)=\O$, $\deg Q=t=\ell(G)$, $v_Q(G)=0$,
$\ell(G-Q)=0$ and $\deg G=t+g-1$.

By Lemmas \ref{ctb1}, \ref{cvaff1} and \ref{cvaff2} we have
\begin{eqnarray*}
\nu^\cP_{\FF_q}(d,\FF_{q^t},1-d(t+g-1)/s)&\le&
\nu^\cP_{\FF_q}(d,\L(G),1-d(t+g-1)/s)\cdot
\nu^\cP_{\L(G)}(d,\FF_{q^t},1)
\\
&\le& \nu^\cP_{\FF_q}(d,\L(G),1-d(t+g-1)/s) \le s.
\end{eqnarray*}
\end{proof}

We are now ready to give the construction.

A {\it tower} of function fields over $\FF_q$ is a sequence
$\cF=(F^{(0)},F^{(1)},F^{(2)},\cdots)$ of function fields
$F^{(i)}/\FF_q$ with $F^{(0)}\subseteq F^{(1)}\subseteq
F^{(2)}\subseteq \cdots$ where each extension $F^{(k+1)}/F^{(k)}$
is finite and separable.

There are many explicit towers known from the literature. We will
use the following ${\cal W}_1$ tower defined in~\cite{GS96}. See
also \cite{GS07} Chapter 1 and \cite{S00} Chapter I. To avoid
confusion we must note here that $F^{(k)}$ here is the function
field $F_{k-1}$ in \cite{GS07,S00}.

\begin{lemma} \labell{tower2}
Let $x_1$ be indeterminate over $\FF_{q^2}$ and
$F^{(1)}=\FF_{q^2}(x_1)$. For $k\ge 2$ let
$F^{(k)}=F^{(k-1)}(x_{k})$ where
$$x_k^q+x_k=\frac{x_{k-1}^q}{x_{k-1}^{q-1}+1}.$$
Let $g_k$ be the genus of $F^{(k)}/\FF_{q^2}$ and $N_k$ the number
of places in $F^{(k)}/\FF_{q^2}$ of degree $1$. Then
\begin{eqnarray}\labell{gk1}
g_k=\left\{ \begin{array}{ll}
q^k-2q^{k/2}+1&\mbox{{\rm if}}\ k=0\mod 2\\
q^k-q^{(k+1)/2}-q^{(k-1)/2}+1 &\mbox{{\rm if}}\ k=1\mod 2
\end{array}\right. ,\end{eqnarray}

\begin{eqnarray}\labell{gk2}
N_k=\left\{ \begin{array}{ll}
(q^2-q)q^{k-1}+2q&\mbox{{\rm if}}\ k\ge 3,q=1\mod 2\\
(q^2-q)q^{k-1}+2q^2 &\mbox{{\rm if}}\ k\ge 3, q=0\mod 2
\end{array}\right.\end{eqnarray} and $N_k\ge (q^2-q)q^{k-1}$ for
$k=1,2$.
\end{lemma}

We are now ready to prove Theorem \ref{T2}, \ref{T2n}
and~\ref{T3n}. We start with Theorem~\ref{T2n}. The proof will be
consequence of the following lemmas. See the table in
Theorem~\ref{T2n}.

Lemmas~\ref{LT01}--\ref{LT05} below prove Theorem~\ref{T2n}.

\begin{lemma}\labell{LT01} Let $k$ be any integer, $Q=q^2$ and $c\ge (k+4)/q^{k}$ be
any constant such that $t:=cq^k$ is an integer. For every
$\epsilon$ such that
$$1>\epsilon\ge \epsilon_{min}:= (c+1)\frac{d}{\sqrt{Q}-1}$$
we have
$$\nu^\P_{\FF_{Q}}(d,\FF_{Q^t},\bepsilon)\le
\left(1+\frac{1}{c}\right)\frac{dt}{\epsilon}.$$
\end{lemma}
\begin{proof} Consider the tower defined in Lemma~\ref{tower2} and the function field
$F^{(k)}/\FF_{q^2}$. The number of places of degree $1$ is at
least $N=q^{k+1}-q^k$ and the genus is $g_k\le
q^k-2q^{k/2}+1=t/c-2\sqrt{t/c}+1$. We now use Lemma~\ref{cvaff3}.
Since $t=cq^k\ge k+4$ and $3+2\log_{q^2}(2g_k+1)\le 4+k$ the first
condition in Lemma~\ref{cvaff3} holds. Therefore, for any $N\ge
s> d(t+g_k-1)$ and $\epsilon=d(t+g_k-1)/s$ we have
\begin{eqnarray*}
\nu^\P_{\FF_{q^2}}(d,\FF_{(q^2)^t},\bepsilon)&\le &
\frac{d(t+g_k-1)}{\epsilon}\\
&\le& \frac{d(t+t/c)-2d\sqrt{t/c}}{\epsilon}\\
&\le &\left(1+\frac{1}{c}\right)\frac{dt}{\epsilon}.
\end{eqnarray*}
The minimal possible $\epsilon$ is
\begin{eqnarray*}
\frac{d(t+g_k-1)}{N}&\le&
\frac{d(cq^k+q^k-2q^{k/2})}{q^{k+1}-q^k}\\
&\le& \frac{d(c+1-2\sqrt{c/t})}{q-1}\\
&\le &(c+1)\frac{d}{q-1}=\epsilon_{min}.
\end{eqnarray*}
\end{proof}

Although the result in Lemma~\ref{LT01} seems to be true for any
perfect square $Q$, the condition $1>\epsilon \ge
(c+1)d/(\sqrt{Q}-1)$ makes sense only when $(c+1)d/(\sqrt{Q}-1)<1$
and therefore $Q>((c+1)d+1)^2$. Therefore we will ignore the
condition on $q$ in the subsequent results with the understanding
that for some $q$ the results are true as the statement is void.

We note here that many other results can be obtained using
different other towers. This will not be discussed in this paper.

We now prove
\begin{lemma}\labell{LT02} Let $q\ge d+1$. Let $k$ be any integer and $c\ge (k+4)/q^k$ be
any constant such that $t:=cq^k$ is an integer. Then for any
$\epsilon$ such that
$$1>\epsilon\ge \epsilon_{min}:= 2(c+1)\frac{d}{q-1}$$
we have
$$\nu^\P_{\FF_{q}}(d,\FF_{q^t},\bepsilon)\le
5\left(1+\frac{1}{c}\right)\left(\frac{d}{\epsilon}\right)^2t.$$
\end{lemma}
\begin{proof} By Lemma~\ref{Triv}, (\ref{ctb11pl2}),
Corollary~\ref{Cqdt} and Lemma~\ref{LT01} we have
\begin{eqnarray*}
\nu^\cP_{\FF_{q}} (d,\FF_{q^{t}},\bepsilon) &\le&
\nu^\cP_{\FF_{q}}(d,\FF_{q^{2t}},\bepsilon)\\
&\le&\nu^\cP_{\FF_{q}}(d,\FF_{q^{2}},\overline{\epsilon/2})\cdot
\nu^\cP_{\FF_{q^2}}(d,\FF_{q^{2t}},\overline{\epsilon/2})\\
&\le&\left(\frac{2d}{\epsilon}+1\right)\left(1+\frac{1}{c}\right)\frac{2dt}{\epsilon}\\
&\le&5\left(1+\frac{1}{c}\right)\left(\frac{d}{\epsilon}\right)^2t.
\end{eqnarray*}
\end{proof}

\begin{lemma}\labell{LT03} Let $q\ge d+1$ and $t\ge 8$. Then for any
$\epsilon$ such that
$$1>\epsilon\ge \epsilon_{min}:= 8\frac{d}{q-1}$$
we have
$$\nu^\P_{\FF_{q}}(d,\FF_{q^t},\bepsilon)\le
120\left(\frac{d}{\epsilon}\right)^3t.$$
\end{lemma}
\begin{proof} Let $k$ be an integer such that $q^k\le t< q^{k+1}$ and $r=q^k$.
Let $\epsilon\ge \epsilon_{min}=8d/(q-1)$. Then,
since by (\ref{Nqk})
$$r\cdot N_q(r)\ge q^{q^k-1} \ge q^{k+2}\ge qt\ge \frac{dt-d+1}{\epsilon/2}$$
by Lemma~\ref{prerec} and \ref{LT02} we have
\begin{eqnarray*}
\nu_{\FF_q}^\P(d,\FF_{q^t},\bepsilon)&\le& \left\lceil
\frac{dt-d+1}{(\epsilon/2)r}\right\rceil \cdot
\nu_{\FF_q}^\P(d,\FF_{q^r},\overline{\epsilon/2})\\
&\le& \left(\frac{3dt}{\epsilon r}\right)
\left(5\cdot 2\cdot \left(\frac{d}{\epsilon/2}\right)^2r\right)\\
&\le& 120 \left(\frac{d}{\epsilon}\right)^3t.
\end{eqnarray*}
\end{proof}

\begin{lemma}\labell{LT04} Let $q\ge d+1$, $8>c>1+8q/(q^2-1)$
and $t\ge 8$. Then for any
$\epsilon$ such that
$$8\frac{d}{q}>\epsilon= c\frac{d}{q}> \frac{d}{q}+\frac{8d}{q^2-1}$$
we have
$$\nu^\P_{\FF_{q}}(d,\FF_{q^t},\bepsilon)\le
O\left(\frac{1}{(c-1)^3}\left(\frac{d}{\epsilon}\right)^4\right)\cdot
t.$$
\end{lemma}
\begin{proof} Let $\epsilon_1=d/q$ and
$\epsilon_2=(c-1)d/q$. By Lemma~\ref{Triv}, (\ref{ctb11pl2}),
Corollary~\ref{Cqdt} and Lemma~\ref{LT03} we have
\begin{eqnarray*}
\nu^\cP_{\FF_{q}} (d,\FF_{q^{t}},\bepsilon)
&\le&
\nu^\cP_{\FF_{q}}(d,\FF_{q^{2t}},\overline{\epsilon_1+\epsilon_2})\\
&\le&\nu^\cP_{\FF_{q}}(d,\FF_{q^{2}},\overline{\epsilon_1})\cdot
\nu^\cP_{\FF_{q^2}}(d,\FF_{q^{2t}},\overline{\epsilon_2})\\
&\le&q\cdot 120\cdot \left(\frac{d}{\epsilon_2}\right)^3t\\
&\le&120
\frac{q^4}{(c-1)^3}t=O\left(\frac{1}{(c-1)^3}\left(\frac{d}{\epsilon}\right)^4\right)\cdot
t.
\end{eqnarray*}
\end{proof}

\begin{lemma}\labell{LT05} Let $q\ge d+1$, $8>c>8q^2/(q^4-1)$
and $t\ge 8$. Then for any $\epsilon$ such that
$$\frac{d}{q}+\frac{9d}{q^2}>\epsilon= \frac{d}{q}+\frac{d}{q^2}+c\frac{d}{q^2}>
\frac{d}{q}+\frac{d}{q^2}+\frac{8d}{q^4-1}$$
we have
$$\nu^\P_{\FF_{q}}(d,\FF_{q^t},\bepsilon)\le
O\left(\frac{1}{c^3}\left(\frac{d}{\epsilon}\right)^9\right)\cdot
t.$$
\end{lemma}
\begin{proof} Let $\epsilon_1=d/q$,
$\epsilon_2=d/q^2$ and $\epsilon_3=cd/q^2$. By Lemma~\ref{Triv},
(\ref{ctb11pl2}), Corollary~\ref{Cqdt} and Lemma~\ref{LT03} we
have
\begin{eqnarray*}
\nu^\cP_{\FF_{q}}
(d,\FF_{q^{t}},\bepsilon) &\le&
\nu^\cP_{\FF_{q}}(d,\FF_{q^{4t}},\overline{\epsilon_1+\epsilon_2+\epsilon_3})\\
&\le&\nu^\cP_{\FF_{q}}(d,\FF_{q^{2}},\overline{\epsilon_1})\cdot
\nu^\cP_{\FF_{q^2}}(d,\FF_{q^{4}},\overline{\epsilon_2}) \cdot
\nu^\cP_{\FF_{q^4}}(d,\FF_{q^{4t}},\overline{\epsilon_3})\\
&\le&q\cdot q^2\cdot 120 \left(\frac{d}{\epsilon_3}\right)^3t\\
&\le&120
\frac{q^9}{c^3}t=O\left(\frac{1}{c^3}\left(\frac{d}{\epsilon}\right)^9\right)\cdot
t.
\end{eqnarray*}
\end{proof}

Lemmas~\ref{LT01}--\ref{LT05} prove Theorem~\ref{T2n}. We now prove Theorem~\ref{T2}.
We show
\begin{lemma}\labell{LT06}
Let $q\ge d+1$, $m$ is any integer, $8>c>8q^{2^{m}}/(q^{2^{m+1}}-1)$
and $t\ge 8$. Then for any constant $m$ and $\epsilon$ such that
$$\frac{d}{q}+\frac{d}{q^2}+\cdots+\frac{d}{q^{2^{m-1}}}+
\frac{9d}{q^{2^{m}}}>\epsilon= \frac{d}{q}+\frac{d}{q^2}+\cdots +\frac{d}{q^{2^m}}+c\frac{d}{q^{2^m}}>
\frac{d}{q}+\frac{d}{q^2}+\cdots+\frac{d}{q^{2^m}}+\frac{8d}{q^{2^{m+1}}-1}$$
we have
$$\nu^\P_{\FF_{q}}(d,\FF_{q^t},\bepsilon)\le
O\left(\frac{1}{c^3}\left(\frac{d}{\epsilon}\right)^{5\cdot 2^m-1}\right)\cdot
t.$$
\end{lemma}
\begin{proof} Let $\epsilon_i=d/q^{2^{i}}$,
 $i=0,1,2,\ldots,m$ and $\epsilon_{m+1}=cd/q^{2^{m}}$. By Lemma~\ref{Triv},
(\ref{ctb11pl2}), {\it 1} in Corollary~\ref{Cqdt} and Lemma~\ref{LT03} we
have
\begin{eqnarray*}
\nu^\cP_{\FF_{q}}
(d,\FF_{q^{t}},\bepsilon) &\le&
\nu^\cP_{\FF_{q}}\left(d,\FF_{q^{2^{m+1}t}},\overline{\sum_{i=0}^{m+1}\epsilon_i}\right)\\
&\le&\left(\prod_{i=0}^m\nu^\cP_{\FF_{q^{2^i}}}(d,\FF_{q^{2^{i+1}}},\overline{\epsilon_i})\right)\cdot
\nu^\cP_{\FF_{q^{m+1}}}(d,\FF_{q^{2^{m+1}t}},\overline{\epsilon_{m+1}})\\
&\le&q\cdot q^2\cdots q^{2^m}\cdot 120 \left(\frac{d}{\epsilon_{m+1}}\right)^3t\\
&\le&120
\frac{q^{5\cdot 2^m-1}}{c^3}t=O\left(\frac{1}{c^3}\left(\frac{d}{\epsilon}\right)^{5\cdot 2^m-1}\right)\cdot
t.
\end{eqnarray*}
\end{proof}

The same proof but using {\it 2} in Corollary~\ref{Cqdt} instead of {\it 1} gives
\begin{lemma}\labell{LT06HP}
Let $q\ge d$, $m$ is any integer, $8>c>8q^{2^{m}}/(q^{2^{m+1}}-1)$
and $t\ge 8$. Then for any constant $m$ and $\epsilon$ such that
$$\epsilon= \frac{d}{q+1}+\frac{d}{q^2+1}+\cdots +\frac{d}{q^{2^m}+1}+c\frac{d}{q^{2^m}}$$
we have
$$\nu^\HP_{\FF_{q}}(d,\FF_{q^t},\bepsilon)\le
O\left(\frac{1}{c^3}\left(\frac{d}{\epsilon}\right)^{5\cdot 2^m-1}\right)\cdot
t.$$
\end{lemma}

We now prove Theorem~\ref{T3n}.

\begin{lemma}\label{LT06} Let $q= d+\delta$ where
$1\le \delta\le 9d$. Then for any $c<1$ and every $\epsilon$ such
that
$$1>\epsilon=\frac{d}{q}+(1-c)\frac{\delta}{q}=1-\frac{c\cdot\delta}{q}
\ge \epsilon_{min}$$ where
$$\epsilon_{min}:=\frac{d}{q}+\frac{12\delta}{q^2}-\frac{12\delta^2}{q^3}=
1-\frac{\delta}{q}+O\left(\frac{\delta}{q^2}\right)$$ we have
$$\nu^\P_{\FF_{q}}(d,\FF_{q^t},\bepsilon)\le
O\left(\frac{d^{\tau+1}}{(1-c)^\tau}\right)\cdot t$$
\end{lemma}
where $\tau=2$ for infinite number of integers $t$ and $\tau=3$ for all integers $t$.
\begin{proof} Let $\epsilon_1=d/q$ and
$$\epsilon_2:=\frac{q\epsilon-d}{q-d}=1-c.$$ Then it is easy to see that
$\bepsilon=\bepsilon_1\bepsilon_2$ and
$$\epsilon_2=\frac{q\epsilon-d}{q-d}\ge
\frac{q\epsilon_{min}-d}{q-d}=\frac{12 d}{q^2}.$$
By Lemma~\ref{Triv}, Corollary~\ref{ctb11} and~\ref{Cqdt} and
Lemma~\ref{LT02} and~\ref{LT03} we have
\begin{eqnarray*}
\nu^\cP_{\FF_{q}} (d,\FF_{q^{t}},\bepsilon) &\le&
\nu^\cP_{\FF_{q}}(d,\FF_{q^{2t}},\bepsilon_1\bepsilon_2)\\
&\le&\nu^\cP_{\FF_{q}}(d,\FF_{q^{2}},\bepsilon_1)\cdot
\nu^\cP_{\FF_{q^2}}(d,\FF_{q^{2t}},\bepsilon_2)\\
&=&q \cdot \nu^\cP_{\FF_{q^2}}(d,\FF_{q^{2t}},\bepsilon_2)\\
&\le& O\left(q\left(\frac{d}{\epsilon_2}\right)^\tau\right)\cdot t\\
&\le& O\left(\frac{d^{\tau+1}}{(1-c)^\tau}\right)\cdot t
\end{eqnarray*}
\end{proof}

The same proof as above (replace each occurrence of $q$ to $q+1$)
gives Theorem~\ref{T3nHP}.

Since all the above bounds use the componentwise, linear, reducible and symmetric
testers that are constructed in Lemma~\ref{cvaff1},  \ref{cvaff2} and Corollary~\ref{Cqdt},
by Lemma~\ref{CLRS} and \ref{ctb1}, Theorem~\ref{RRR} follows.

\subsection{Testers for Small Fields}\label{SField}

In this section we use our
results from the previous sections to construct testers for small
fields. We give constructions for testers for $\HLF(\FF_q,n,d)$
from $\FF_{q^t}$ to $\FF_q$ for any $q$. Theorem~\ref{TRIVLB} and  Theorem~\ref{lowerB}
show that the size
of such tester is at least $(1+1/(q-1))^dt$ and its density is at most $\bepsilon\le (1-1/q)^d$.
One of the testers we give in
this subsection is a tester of size $(1+(\log q)/q)^dt$ and density $\bepsilon\le (1-(\log q)/q)^d$.

We first prove

\begin{theorem} \labell{T1} Let $q<d+1$ be a power of prime
and $t$ be any integer.
Let $r$ be an integer such that $q^{2^{r-1}}<9d\le q^{2^r}$. Let
$\bfepsilon=(\epsilon_0,\ldots,\epsilon_{r-1},\epsilon_r)$ where
$\epsilon_i(q^{2^i}+1)\le q^{2^i}$ is an integer for
$i=0,1,\ldots,r-1$ and $2/3\ge \epsilon_r\ge 1/3$. Let
$$c_{q,\bfepsilon}:=
\sum_{i=0}^{r-1} \frac{\log(q^{2^i}+1)}{\epsilon_i(q^{2^i}+1)}$$
and
$$\pi_{q,\bfepsilon}:=\sum_{i=0}^{r-1}
\frac{-\log(1-\epsilon_i)}{\epsilon_i(q^{2^i}+1)}.$$

Then
$$\overline{\epsilon^\star}:=\bepsilon_r\prod_{i=0}^{r-1}\bepsilon_i^{\lceil
d/(\epsilon_i(q^{2^i}+1))\rceil}\ge
\frac{2^{-\pi_{q,\bfepsilon}\cdot d}}{\Theta(d^2)},$$ and
$$\nu_{\FF_q}(d,\FF_{q^t},\overline{\epsilon^\star})
\le \Theta(d^5)\cdot 2^{c_{q,\bfepsilon}\cdot d}\cdot t.$$
\end{theorem}
\begin{proof} By Lemma \ref{Triv},
Corollary \ref{ctb11}, Lemma \ref{basn} and \ref{qdt} we have
\begin{enumerate}
\item \label{P0} $\nu_{\FF_q}(d,\FF_{q^{t_1}},\bepsilon)\le
\nu_{\FF_q}(d,{\FF_{q^{t_1t_2}}},\bepsilon)$.

\item\label{P1}
$\nu_{\FF_q}(d,\FF_{q^{t_1t_2}},\bepsilon_1\bepsilon_2)\le
\nu_{\FF_q}(d,{\FF_{q^{t_1}}},\bepsilon_1)\cdot
\nu_{\FF_{q^{t_1}}}(d,\FF_{q^{t_1t_2}},\bepsilon_2).$

\item\label{P2}
$\nu_{\FF_q}(d_1+d_2,\FF_{q^t},\bepsilon_1\bepsilon_2)\le
\nu_{\FF_q}(d_1,\FF_{q^t},\bepsilon_1)\cdot
\nu_{\FF_q}(d_2,\FF_{q^t},\bepsilon_2).$

\item\label{P3} $\nu_{\FF_q}(d,\FF_{q^2},\bepsilon)\le q+1$ for any $\epsilon<1$ such that
$\epsilon(q+1)$ is an integer and $d\le \epsilon(q+1)$.
\end{enumerate}

Let $\eta_i=\epsilon_i(q^{2^i}+1)$. Then
\begin{eqnarray*}
\nu_{\FF_q}(d,\FF_{q^t},\overline{\epsilon^\star})&\le &
\nu_{\FF_q}(d,\FF_{q^{t\cdot 2^r}},\overline{\epsilon^\star})\ \ \
\mbox{ By (\ref{P0}.)}\\
&\le & \left(\prod_{i=0}^{r-1}\nu_{\FF_{q^{2^i}}}
\left(d,\FF_{q^{2^{i+1}}},\bepsilon_i^{\lceil
d/\eta_i\rceil}\right) \right)
\nu_{\FF_{q^{2^r}}}(d,\FF_{(q^{2^r})^t},\bepsilon_r)\ \mbox{\ \ By (\ref{P1}.)} \\
&\le &\left(\prod_{i=0}^{r-1}
\nu_{\FF_{q^{2^i}}}(\eta_i,\FF_{q^{2^{i+1}}},\bepsilon_i)^{\lfloor
d/\eta_i\rfloor}
\cdot \nu_{\FF_{q^{2^i}}}(d-\eta_i{\lfloor
d/\eta_i\rfloor},\FF_{q^{2^{i+1}}},\bepsilon_i)\right)\\
& & \ \ \ \ \ \ \ \ \cdot \nu_{\FF_{q^{2^r}}}(d,\FF_{(q^{2^r})^t},
\bepsilon_r)\
\mbox{\ \ By (\ref{P2}.)}\\
&\le& \left(\prod_{i=0}^{r-1}\left(q^{2^i}+1\right)^{\lceil
d/\eta_i\rceil}\right)
\nu_{\FF_{q^{2^r}}}(d,\FF_{(q^{2^r})^t},\bepsilon_r)
\mbox{\ \ By (\ref{P3}.)}\\
&\le&\left(\left(\prod_{i=0}^{r-1}\left(q^{2^i}+1\right)\right)
2^{c_{q,\bfepsilon}\cdot
d}\right)\nu_{\FF_{q^{2^r}}}(d,\FF_{(q^{2^r})^t},\bepsilon_r) \\
&\le& \frac{q^{2^r}-1}{q-1}
\nu_{\FF_{q^{2^r}}}(d,\FF_{(q^{2^r})^t},\bepsilon_r)\cdot 2^{c_{q,\bfepsilon}\cdot d}\\
&\le& \Theta(d^5)\cdot 2^{c_{q,\bfepsilon}\cdot d }\cdot t\mbox{\
\ By Lemma~\ref{LT03}}
\end{eqnarray*}

Now
\begin{eqnarray*}
\overline{\epsilon^\star}&=&\bepsilon_r\prod_{i=0}^{r-1}\bepsilon_i^{\lceil
d/(\epsilon_i(q^{2^i}+1))\rceil}\\
&\ge&\bepsilon_r\left( \prod_{i=0}^{r-1}\bepsilon_i\right)
\prod_{i=0}^{r-1}\bepsilon_i^{ d/(\epsilon_i(q^{2^i}+1))}\\
&\ge&\frac{1}{3}\left( \prod_{i=0}^{r-1}\frac{1}{q^{2^i}+1}\right)
\left(\prod_{i=0}^{r-1}(\bepsilon_i)^{ 1/(\epsilon_i(q^{2^i}+1))}\right)^d\\
&=&\frac{1}{3}\frac{q-1}{q^{2^r}-1} 2^{-\pi_{q,\bfepsilon}\cdot
d}\ge \frac{2^{-\pi_{q,\bfepsilon}\cdot d}}{\Theta(d^2)}.
\end{eqnarray*}
\end{proof}

Proposition \ref{Prop1} in Appendix B will help us choose $\epsilon_i$
in Theorem~\ref{T1} to obtain different results. We first prove

\begin{corollary} \labell{Co1} Let  $q<d+1$ be a power of prime. For any integer $m$
such that $1\le m\le q$ we have: For
\begin{eqnarray*}
\overline{\epsilon^\star}&=&
\left(1-\frac{m}{q+1}\right)^{\frac{1}{m}\left(1+
\frac{1}{\Theta(q)}\right)d}\\
\end{eqnarray*}
we have
$$\nu_{\FF_q}(d,\FF_{q^t},\overline{\epsilon^\star})
\le (q+1)^{\frac{d}{m}\left(1+\frac{1}{\Theta(q)}\right)}\cdot
t.$$

The following Table shows the results for different choices of $m$
(ignoring the small terms)
\begin{center}
\begin{tabular}{|l|l|c|}
$m$ & $\overline{\epsilon^\star}$ & $\nu_{\FF_q}(d,\FF_{q^t},\overline{\epsilon^\star})/t$ \\
\hline\hline $m=1$ & $\left(1-\frac{1}{q+1}\right)^d$ & $(q+1)^d$\\
\hline $m=\log(q+1)/c$, $c=o(\log(q+1))$ & $\left(1-\frac{1}{q+1}\right)^d$ & $2^{cd}$\\
\hline $m=o(q)$,\ $\omega(\log(q+1))$ &
$\left(1-\frac{1}{q+1}\right)^d$ &
$\left(1+\frac{\ln(q+1)}{m}\right)^d$\\
\hline $m=c(q+1)$, $c<1$, $c=\Theta(1)$ &
$\left(1-\frac{\ln(1/(1-c))}{c(q+1)}\right)^d$ &
$\left(1+\frac{\ln(q+1)}{c(q+1)}\right)^d$\\
\hline $m=(q+1)-(q+1)/c$, $c=\omega(1)$ & $\left(1-\frac{\ln
c}{q+1}\right)^d$ &
$\left(1+\frac{\ln(q+1)}{q+1}\right)^d$\\
\hline $m=(q+1)-c$, $c=\Theta(1)$ &
$\left(1-\frac{\ln(q+1)}{q+1}\right)^d$ &
$\left(1+\frac{\ln(q+1)}{q+1}\right)^d$\\
 \hline
\end{tabular}
\end{center}

%
\end{corollary}
\begin{proof} We use Theorem \ref{T1} and
Proposition \ref{Prop1} in Appendix B. We choose
$\epsilon_i(q+1)=m$, for $i=0,1,2,\ldots,r-1$ and
$\epsilon_r=1/3$.
\end{proof}

The reason for the choice of such $\epsilon_i$ in Theorem~\ref{Co1} is explained in
Appendix C.

The following corollary gives the minimal possible size of a
tester that can be obtained from Theorem~\ref{T1}

\begin{corollary} \labell{density2}
Let  $q<d+1$. Let
$$c_q=\sum_{i=0}^\infty \frac{\log(q^{2^i}+1)}{q^{2^i}}=\Theta\left(\frac{\log q}{q}\right).$$ For
$$\overline{\epsilon^\star}=2^{-c_qd}/\Theta(d^2)$$ we have
$$\nu_{\FF_q}(d,\FF_{q^t},\overline{\epsilon^\star})
\le \Theta(d^5)\cdot 2^{c_qd}\cdot t$$

In particular we have following values of $c_q$
\begin{center}
\begin{tabular}{|l|c|c|}
\hline $q$&$c_q$\\
\hline \hline $2$&$1.659945821$
\\
\hline $3$&$1.116191294$\\
\hline $4$&$0.867464571$\\
\hline $5$&$0.719921672$\\
\hline $7$&$0.548433289$\\
\hline
\end{tabular}\end{center}
\end{corollary}
\begin{proof} We use Theorem \ref{T1}. We choose
$\epsilon_i(q^{2^i}+1)=q^{2^i}$ for $i=0,1,\ldots,r-1$ and
$\epsilon_r=1/3$.
\end{proof}

In Theorem \ref{lowerB} we have shown that there is no
$(\HLF(\FF_q, n, d),\FF_{q^t},\FF_q)$-$\bepsilon$-tester of
density greater than $\bepsilon_{min}=(1-1/q)^d$. We now use Theorem~\ref{T1} to show that one can get a tester with density
$\bepsilon=(1-1/q-1/poly(q))^d$ and size $q^{O(d)}\cdot t$.

\begin{corollary} \labell{maxdense} Let  $q<d+1$.
For every $(\log d)/d\le \delta\le 1/q^2$ we have: For
$$\overline{\epsilon^\star}= \left(1-\frac{1}{q}-
c_1\delta \right)^d,$$
$$\nu_{\FF_q}(d,\FF_{q^t},\overline{\epsilon^\star})
\le \left(\frac{c_2\log_q (1/\delta)}{\delta}\right)^{d}\cdot t$$
for some constants $c_1$ and $c_2$.

In particular, for
$$\overline{\epsilon^\star}= \left(1-\frac{1}{q}-
\frac{1}{poly(q)} \right)^d,$$
$$\nu_{\FF_q}(d,\FF_{q^t},\overline{\epsilon^\star})
\le q^{O(d)}\cdot t.$$
\end{corollary}
\begin{proof} Consider the integer $r$ in Theorem~\ref{T1}.
Let $k<r$ be constant that will be determined later such that
$(\log d)/d<2^k/q^{2^{k+1}}$. Apply Theorem~\ref{T1} and consider
the case where $\epsilon_i(q^{2^i}+1)=1$ for $i=0,\ldots,k-1$,
$m_k=\epsilon_k(q^{2^{k}}+1)=2^k=\lfloor \log(q^{2^k}+1)/\log
q\rfloor$, $\epsilon_i=1/2$ for $i\ge k+1$ and $\epsilon_r=2/3$.
Then by Theorem \ref{T1} and Proposition \ref{Prop1} in Appendix
B,
\begin{eqnarray*}
\overline{\epsilon^\star}^{1/d}&\ge&\left(\frac{1}{\Theta(d^2)}\right)^{1/d}\cdot
\left(\prod_{i=0}^{k-1}\left(1-\frac{1}{q^{2^{i}}+1}\right)\right)\cdot
\left(1-\frac{m_k}{q^{2^k}+1}\right)^{\frac{1}{m_k}}
\cdot\prod_{i=k+1}^r
\left(\frac{1}{2}\right)^{\frac{2}{q^{2^i}+1}}\\
&=&\left(\frac{1}{\Theta(d^2)}\right)^{1/d}\cdot
\left(1-\frac{1}{q}\right)\left(1-\frac{1}{q^{2^k}}\right)^{-1}\cdot
\left(1-\frac{1}{q^{2^k}+1}-\Theta\left(\frac{2^k}{q^{2^{k+1}}}\right)\right)
\cdot \left(1-\Theta\left(\frac{1}{q^{2^{k+1}}}\right)\right)\\
&=&\left(\frac{1}{\Theta(d^2)}\right)^{1/d}\left(1-\frac{1}{q}\right)
\left(1+\frac{1}{q^{2^k}}+\Theta\left(\frac{1}{q^{2^{k+1}}}\right)\right)
\left(1-\frac{1}{q^{2^k}+1}-\Theta\left(\frac{2^k}{q^{2^{k+1}}}\right)\right)\\
&=&\left(1-\Theta\left(\frac{\log
d}{d}\right)\right)\left(1-\frac{1}{q}\right)
\left(1-\Theta\left(\frac{2^k}{q^{2^{k+1}}}\right)\right)\\ & =&
1-\frac{1}{q}-\Theta\left(\frac{2^k}{q^{2^{k+1}}}\right)
\end{eqnarray*} Denote the small term $W_k=2^k/q^{2^{k+1}}$
and choose $k$ such that $W_k\le \delta$ and $W_{k-1}>\delta$.
Then by Theorem~\ref{T1} and Proposition \ref{Prop1} in Appendix
B,
\begin{eqnarray*}
\nu_{\FF_q}(d,\FF_{q^t},\overline{\epsilon^\star})^{1/d} &\le&
(\Theta(d^5))^{1/d}\left(\prod_{i=0}^{k-1} (q^{2^i}+1)\right)
(q^{2^k}+1)^{\frac{1}{m_k}}
\prod_{i=k+1}^r (q^{2^i}+1)^{\frac{2}{q^{2^i}+1}}\\
&\le & \frac{q^{2^{k}}-1}{q-1} \Theta(q) = \Theta(q^{2^k})=
\Theta\left(\frac{2^{k-1}}{W_{k-1}}\right)\\ &=&
\Theta\left(\frac{\log_q (1/W_{k-1})}{W_{k-1}}\right)=
O\left(\frac{\log_q(1/\delta)}{\delta}\right).
\end{eqnarray*}
\end{proof}

The last result in this subsection is
\begin{theorem} All the above testers are componentwise and
linear but not reducible and not symmetric.
\end{theorem}
\begin{proof} All the testers built in the previous sections and subsections are
componentwise and linear and since all the constructions used in Theorem~\ref{T1}
preserve those two
properties, the testers in this subsection are componentwise and linear.

The construction in Theorem~\ref{T1} uses the tester constructed in~{\it 3} of Lemma~\ref{qdt} which
is not reducible ($l_\infty(1)=0$). It also uses construction
{\it \ref{basn2}} of Lemma~\ref{basn}
that, by Lemma~\ref{psym}, does not preserve the symmetric property.
\end{proof}

\section{Almost Linear Time Constructions and Locally Explicit}\labell{s18}
In this section we show that a dense tester $(\P(\FF_{q},n,d),\FF_{q},\FF)$-$\bepsilon$-tester
of size $s=poly(d/\epsilon)\cdot t$ can be constructed in almost linear time
in $s$ and $p$ and is locally explicit. Here $p$ is
the characteristic of the field which is $O(1)$ for all the applications we have in~\cite{B1}.

\subsection{Dense Testers for Very Small $t$ and Large $q$}\labell{s19}
In this section give linear time constructions for small $t$.

In Theorem~\ref{TRIVLB} we showed that the size of any
$(\P(\FF_q,n,d),\FF_{q^t},\FF_q)$-$\bepsilon$-tester is at least
$\Omega((d/\epsilon)\cdot t)$. In Theorem~\ref{lowerB} we showed
that the best possible density one can get for
$(\P(\FF_q,n,d),\FF_{q^t},\FF_q)$-$\bepsilon$-tester is
$\epsilon\ge d/q$. In this section we show that for small $t=o(q)$
one can in almost linear time build testers of size
$\poly(d/\epsilon)\cdot t$ of density $d/q+o(d/q)$.

\noindent
We will abuse the notations $\nu_R^\P, \nu_R^\HP$ or $\nu_R$ and identify every inequality in $\nu_R^\P, \nu_R^\HP$ or $\nu_R$ with its
corresponding construction.
For example, by the first inequality in (\ref{bas1}) below
we mean the following statement: From a $(\P(\FF_{q},n,d),\FF_{q^2},\FF)$-$\bepsilon_2$-tester of size $s_1$
a $(\P(\FF_{q},n,d),\FF_{q^t},\FF)$-$\bepsilon_1\bepsilon_2$-tester of size $s_2:=s_1\lfloor (dt-d+1)/(2\epsilon_1)\rfloor$ can
constructed in almost linear time. See Important Note 1 in Subsection~\ref{s11}.

Note that just reading the elements of the field $\FF_{q^t}$ takes time $t\log q$. Therefore
one cannot expect any time complexity that is better than $\tilde O(t)$.

\begin{theorem}\labell{poly1}
The following
$(\P(\FF_q,n,d),\FF_{q^t},\FF_q)$-$\bepsilon$-tester can be
constructed in deterministic time~$T\cdot poly(\log (qtd/\epsilon))=\tilde O(T)$ and any entry
of any map in the
tester can be constructed and computed in time~$T'\cdot poly(\log (qtd/\epsilon))=\tilde O(T')$.
\begin{center}
\begin{tabular}{l|c|l|l|l|l|}
&Size$=O(\cdot)$ & $\epsilon$& $t$& $T$&$T'$\\
\hline $1)$&$\frac{d}{\epsilon}\cdot t$ & $\epsilon\ge\frac{d(t-1)}{q}$ & \mbox{\rm ANY}& $Size$&$t$\\
\hline $2)$&$\frac{d^2}{\epsilon(\epsilon-d/q)}\cdot t$ & $\epsilon\ge\frac{d}{q}+\frac{dt-d+1}{q^2-q}$ & $<q-1$& $Size+p^{1/2}$&$t+p^{1/2}$\\
\hline $3)$&$\frac{1}{c}\left(\frac{d}{\epsilon}\right)^2\cdot t$ & $\epsilon\ge(1+c)\frac{d}{q}$ & $<c(q-1)$&$Size+p^{1/2}$&$t+p^{1/2}$\\
\hline $4)$&$\left(\frac{d}{\epsilon}\right)^3$ & $\epsilon=\frac{d}{q}+o
\left(\frac{d}{q}\right)$ & $=o(q)$& $Size$&$t+p^{1/2}$\\
\hline $5)$&$\frac{1}{c^2}\left(\frac{d}{\epsilon}\right)^3\cdot t$ & $O\left(\frac{d}{q}\right)=\epsilon\ge (1+c)\frac{d}{q}$ & $\frac{q}{\log q}<t< \frac{c}{2}q^{\frac{c}{2}(q-1)-3}$&$Size$ &$t$\\
\hline
$6)$ & $\left(\frac{d}{\epsilon}\right)^4\left(\log^3\frac{d}{\epsilon}\right)\cdot t$ & $O\left(\frac{d}{q}\right)=\epsilon\ge\frac{d}{q}+o\left(\frac{d}{q}\right)$
& $q^{4c'q/\log q}<t<q^{q^{c'q/\log q}}$ & $Size$&$t$\\
\hline
\end{tabular}
\end{center} for any $1\ge c\ge 0$ and any constant $c'>1$.
\end{theorem}
\begin{proof}
By Corollary~\ref{Cqdt} we have for $\epsilon\ge d(t-1)/q$, a tester of size
$O(({d}/{\epsilon})\cdot t)$ can be constructed in linear time in $dt/\epsilon$
and any entry of any map in the tester can be constructed and computed in time $\tilde O(t)$.
This implies result $1$.

We now prove result $2$. Consider the field $\FF_{q^2}$. Then by (\ref{NqkE}), $2N_q(2)=q^2-q$. By
Lemma~\ref{prerec} and Corollary~\ref{Cqdt}, for any
$$\epsilon_1\ge \frac{dt-d+1}{q^2-q}\mbox{\ and \ } \epsilon_2\ge \frac{d}{q}$$ we have
\begin{eqnarray}\labell{bas1}\nu_{\FF_q}^\P(d,\FF_{q^t},\bepsilon_1\bepsilon_2)\le \left\lceil
\frac{dt-d+1}{2\epsilon_1}\right\rceil
\nu_{\FF_q}^\P(d,\FF_{q^2},\bepsilon_2)\le \left\lceil
\frac{dt-d+1}{2\epsilon_1}\right\rceil
\left\lceil\frac{d}{\epsilon_2}\right\rceil.\end{eqnarray}

Notice that for $t<q-1$,
$$\frac{dt-d+1}{q^2-q}<\frac{d}{q}.$$ We now distinguish between two cases.
When $2d/q\le \epsilon$ we
substitute $\epsilon_1=\epsilon_2=\epsilon/2$ and get
$$\nu_{\FF_q}^\P(d,\FF_{q^t},\bepsilon)\le O\left(\left(\frac{d}{\epsilon}\right)^2\cdot t\right)=
O\left(\left(\frac{d^2}{\epsilon(\epsilon-d/q)}\right)\cdot t\right).$$
When
$$\frac{d}{q}+\frac{dt-d+1}{q^2-q}\le \epsilon<\frac{2d}{q}$$
we substitute $\epsilon_2=d/q$ and $\epsilon_1=\epsilon-\epsilon_2$ and get
$$\nu_{\FF_q}^\P(d,\FF_{q^t},\bepsilon)\le O\left(\left(\frac{d^2}{\epsilon(\epsilon-d/q)}\right)\cdot t\right).$$
The time complexity follows from Lemma~\ref{prerec} and Lemma~\ref{FF1} (for constructing $\FF_{q^2}$).

We now prove result $3$. For $\epsilon\ge (1+c)d/q$ and $t<c(q-1)$
where $1\ge c\ge 0$ is any constant we have $\epsilon \ge d/q+(dt-d+1)/(q^2-q)$ and therefore by $(2)$,
$$\nu_{\FF_q}^\P(d,\FF_{q^t},\bepsilon)\le O\left(\frac{1}{c}\left(\frac{d}{\epsilon}\right)^2\cdot t\right).$$

To prove result $4$, we use (\ref{bas1}) for $\epsilon_2=d/q$ and $\epsilon_1=O(dt/q^2)$. Notice here that $Size=O(q^3)$ which is much larger than
the extra term $\tilde O(p^{1/2})$.

To prove result $5$, we use Lemma~\ref{prerec} with $k=(c/2)(q-1)-1$, $\epsilon_1=(c/2)(d/q)$
and $\epsilon_2=\epsilon-\epsilon_1\ge (1+c/2)(d/q)$. Now
since
$$\epsilon_1= \frac{c}{2}\frac{d}{q}\ge
\frac{dt}{q^{k-1}}\ge \frac{dt-d+1}{q^{k-1}}\ge \frac{dt-d+1}{k\cdot N_q(k)}$$ we get
$$\nu_{\FF_q}^\P(d,\FF_{q^t},\bepsilon)\le \nu_{\FF_q}^\P(d,\FF_{q^t},\bepsilon_1\bepsilon_2)\le \left\lceil \frac{dt-d+1}{\epsilon_1k}\right\rceil
\nu_{\FF_q}^\P(d,\FF_{q^k},\bepsilon_2).$$
By result $3$ we have $\nu_{\FF_q}^\P(d,\FF_{q^k},\bepsilon_2)= O((1/c)(d/\epsilon_2)^2k)$ and therefore
$$\nu_{\FF_q}^\P(d,\FF_{q^t},\bepsilon)=
O\left(\frac{1}{c^2}\left(\frac{d}{\epsilon}\right)^3\cdot t\right).$$
The time complexity is $O((1/c^2)(d/\epsilon)^3t+q^3p^{1/2}+q^4\log^2q)=\tilde O(Size)$.

To prove result $6$, take $c=\Theta(1/\log{q})$ such that $k:=q^{{c'q/\log q}}+3<(c/2) q^{(c/2)(q-1)-3}$, $\epsilon_1=d/(q\log q)$ and $\epsilon_2=(1+c)(d/q)$. Since
$$kN_q(k)\ge q^{k-1}\ge (\log q)qt \ge \frac{dt-d+1}{\epsilon_1}$$
by Lemma~\ref{prerec} and $(5)$, for $\epsilon=d/q+\Theta(d/(q\log q))$,
\begin{eqnarray*}
\nu_{\FF_q}^\P(d,\FF_{q^t},\bepsilon)&\le & \nu_{\FF_q}^\P(d,\FF_{q^t},\bepsilon_1\bepsilon_2)\le \left\lceil \frac{dt-d+1}{\epsilon_1k}\right\rceil
\nu_{\FF_q}^\P(d,\FF_{q^k},\bepsilon_2)\le O(q(\log q)t\cdot (\log q)^2q^3)\\
&=& O\left(\left(\frac{d}{\epsilon}\right)^4\log^3 \left(\frac{d}{\epsilon}\right)\cdot t\right).
\end{eqnarray*}
\end{proof}

The above Theorem give dense testers for $\epsilon=O(d/q)$. For $\epsilon=\omega(d/q)$ we have
\begin{theorem}\labell{poly2}
The following
$(\P(\FF_q,n,d),\FF_{q^t},\FF_q)$-$\bepsilon$-tester can be
constructed in deterministic time~$T\cdot poly(\log (qtd/\epsilon))=\tilde O(T)$ and any entry in any map in the
tester can be constructed and computed in time~$T'\cdot poly(\log (qtd/\epsilon))=\tilde O(T')$
\begin{center}
\begin{tabular}{l|c|l|l|l|l|}
&Size$=O(\cdot)$ & $\epsilon$& $t$& $T$&$T'$\\
\hline $1)$ &$\frac{d}{\epsilon}\cdot t$ & $\epsilon\ge(t-1)\frac{d}{q}$ & \mbox{\rm ANY}& $Size$& $t$\\
\hline $2)$&$\left(\frac{d}{\epsilon}\right)^2\cdot t$ & $\epsilon\ge 2\eta\frac{d}{q}$ & $\le\eta q^{\eta-1}$& $Size+\eta^3 p^{1/2}+\eta^4$&$t+\eta^3p^{1/2}+\eta^4$\\
\hline $3)$ &$\left(\frac{d}{\epsilon}\right)^3\cdot t$ & $\epsilon\ge 3\eta\frac{d}{q}$ & $q^{4\eta}\le t\le q^{\eta q^{\eta-1}-2}$& $Size$&$t$\\
\hline $4)$ &$\left(\frac{d}{\epsilon}\right)^4\cdot t$ & $\epsilon\ge 4\eta\frac{d}{q}$ & $q^{4\eta q^{\eta-1}}\le t\le q^{q^{\eta q^{\eta-1}-2}-2}$& $Size$&$t$\\ \hline
\end{tabular}
\end{center} where $\eta\le q/d$ is any integer. In particular for $\epsilon\ge 34\cdot d/q$ and any $t\le q^{q^q}$
a $(\P(\FF_q,n,d),\FF_{q^t},\FF_q)$-$\bepsilon$-tester of size
$$S=O\left(\left(\frac{d}{q}\right)^4\cdot t\right)$$ can be
constructed in deterministic time $\tilde O(S+p^{1/2})$ and any entry of any map in the
tester can be constructed and computed in time~$\tilde O(t+p^{1/2})$.
\end{theorem}
\begin{proof}
Result 1 is the same as result 1 in Theorem~\ref{poly1}.

We now prove result $2$. Consider the field $\FF_{q^{\eta+1}}$. Then by (\ref{Nqk}), $(\eta+1) N_q(\eta+1)\ge q^{\eta}$. By
Lemma~\ref{prerec} and Corollary~\ref{Cqdt}, for any
$$\epsilon_1\ge \frac{dt-d+1}{q^\eta}\mbox{\ and \ } \epsilon_2\ge \eta\frac{d}{q}$$ we have
\begin{eqnarray}\labell{bas2}\nu_{\FF_q}^\P(d,\FF_{q^t},\bepsilon_1\bepsilon_2)\le \left\lceil
\frac{dt-d+1}{(\eta+1)\epsilon_1}\right\rceil
\nu_{\FF_q}^\P(d,\FF_{q^{\eta+1}},\bepsilon_2)\le \left\lceil
\frac{dt-d+1}{(\eta+1)\epsilon_1}\right\rceil
\left\lceil\frac{d}{\epsilon_2}\eta\right\rceil.\end{eqnarray}

Notice that for $t\le \eta q^{\eta-1}$,
$$\frac{dt-d+1}{q^\eta}<\eta \frac{d}{q}.$$
When $2\eta d/q\le \epsilon$ we
substitute $\epsilon_1=\epsilon_2=\epsilon/2$ and get
$$\nu_{\FF_q}^\P(d,\FF_{q^t},\bepsilon)\le O\left(\left(\frac{d}{\epsilon}\right)^2\cdot t\right).$$

We now prove result $3$. Consider the field $\FF_{q^{k}}$ where $k=\eta q^{\eta-1}$. Then by (\ref{Nqk}), $k N_q(k)\ge q^{k-1}\ge (dt-d+1)/\epsilon_1$ where $\epsilon_1\ge \eta(d/q)$. Let $\epsilon_2\ge 2\eta(d/q)$. By
Lemma~\ref{prerec} and result $2$ we have
\begin{eqnarray}\labell{bas3}\nu_{\FF_q}^\P(d,\FF_{q^t},\bepsilon_1\bepsilon_2)\le \left\lceil
\frac{dt-d+1}{k\epsilon_1}\right\rceil
\nu_{\FF_q}^\P(d,\FF_{q^{k}},\bepsilon_2)\le \left\lceil
\frac{dt-d+1}{k\epsilon_1}\right\rceil\cdot O\left(\left(
\frac{d}{\epsilon_2}\right)^2 k\right).\end{eqnarray}

When $3\eta d/q\le \epsilon$ we
substitute $\epsilon_1=\epsilon/3$ and $\epsilon_2=2\epsilon/3$ and get
$$\nu_{\FF_q}^\P(d,\FF_{q^t},\bepsilon)\le O\left(\left(\frac{d}{\epsilon}\right)^3\cdot t\right).$$

We now prove result $4$. Consider the field $\FF_{q^{k}}$ where $k=q^{\eta q^{\eta-1}-2}$. Then by (\ref{Nqk}), $k N_q(k)\ge q^{k-1}\ge (dt-d+1)/\epsilon_1$ where $\epsilon_1\ge \eta(d/q)$. Let $\epsilon_2\ge 3\eta(d/q)$. By
Lemma~\ref{prerec} and result $3$ we have
\begin{eqnarray}\labell{bas4}\nu_{\FF_q}^\P(d,\FF_{q^t},\bepsilon_1\bepsilon_2)\le \left\lceil
\frac{dt-d+1}{k\epsilon_1}\right\rceil
\nu_{\FF_q}^\P(d,\FF_{q^{k}},\bepsilon_2)\le \left\lceil
\frac{dt-d+1}{k\epsilon_1}\right\rceil\cdot O\left(\left(
\frac{d}{\epsilon_2}\right)^3 k\right).\end{eqnarray}

When $4\eta d/q\le \epsilon$ we
substitute $\epsilon_1=\epsilon/4$ and $\epsilon_2=3\epsilon/4$ and get
$$\nu_{\FF_q}^\P(d,\FF_{q^t},\bepsilon)\le O\left(\left(\frac{d}{\epsilon}\right)^4\cdot t\right).$$

The final result in the Theorem follows from results $2$, $3$ and $4$ with $\eta=17,4,2$ respectively.
\end{proof}

\subsection{Dense Testers for any $t$ and Large $q$}\labell{s20}

In this section we first prove
\begin{theorem} \labell{LarE}
Let $q\ge d+1$, $c>0$ be a constant and $$\epsilon\ge 34\cdot \frac{d}{q}.$$ A
$(\P(\FF_q,n,d),\FF_{q^t},\FF_q)$-$\bepsilon$-tester of size
$$s=\left(\frac{d}{\epsilon}\right)^4\cdot t$$ can be
constructed in time $T=\tilde O(s+p^{1/2})$ and any entry of any map in the
tester can be constructed and computed in time~$\tilde O(t+p^{1/2})$.
\end{theorem}
\begin{proof} By Theorem~\ref{poly2} we may assume that $t\ge w:=q^{q^{q}}$.
Let $t_1=\lceil \log_q t\rceil +2$ and $t_2=c_1q^k$ where $c_1<1$ is any small constant such that $c_1q^{k-1}$ is an integer and
$c_1q^k\ge \lceil \log_q t_1\rceil +2\ge c_1q^{k-1}$. Since for $\epsilon_1=\epsilon/4$
$$t_1 N_q(t_1)\ge q^{t_1-1}\ge q t \ge \frac{dt-d+1}{\epsilon_1}$$ and
$$t_2 N_q(t_2)\ge q^{t_2-1}\ge q t_1 \ge \frac{dt_1-d+1}{\epsilon_1}$$
by Lemma~\ref{prerec} and Lemma~\ref{LT02}
$$\nu_{\FF_q}^\P(d,\FF_{q^t},\bepsilon)\le \left\lceil\frac{dt}{\epsilon_1t_1}\right\rceil\cdot
\left\lceil\frac{dt_1}{\epsilon_1t_2}\right\rceil\cdot
\nu_{\FF_q}^\P(d,\FF_{q^{t_2}},\overline{\epsilon-2\epsilon_1})
= O\left( \left(\frac{d}{\epsilon}\right)^4\cdot t\right).$$

Now we prove that the above can be constructed in time $T$.
If $t\le w$ then the time complexity follows from Theorem~\ref{poly2}.
Now suppose $t\ge w$. By Lemma~\ref{prerec} the reduction to $\FF_{q^{t_2}}$ can be done
in time $\tilde O(s+t_1^3p^{1/2}+t_1^4)=\tilde O(s)$. By Lemma~\ref{LT02} a symmetric $(\P(\FF_q,n,d),\FF_{q^{t_2}},\FF_q)$-$(\bepsilon-2\epsilon_1)$-tester of size $s=15(d/(\bepsilon-2\epsilon_1))^2t$ exists. We will construct it by exhaustive search.
We exhaustively search for linear maps $L=\{l_1,\ldots,l_s\}$ where $s=15(d/(\epsilon-2\epsilon_1))^2t_2\le q^3\log_q\log_q t\le (\log_q\log_q t)^4$ in $\FF_{q^{t_2}}^*$, and check if every $\lfloor (\epsilon-2\epsilon_1)|L|\rfloor +1$ elements in $L$ is a tester. Verifying whether a set of
maps is a tester can be done in polynomial time in $s$~\cite{B2}. The number of all possible sets $L$ and subsets of $L$ is at most
$${|\FF^*_{q^{t_2}}| \choose s}2^s\le q^{st_2}\le q^{2(\log_q\log_qt)^6}.$$
Now notice that $q\ge 34d/\epsilon\ge 68$ and since $2(\log_q\log_q t)^6<\log_qt$ for $t\ge q^{q^q}$ and  $q\ge 68$ we have $q^{2(\log_q\log_qt)^6}<t$. Therefore
the time complexity of the exhaustive search is less than $s$. This finishes the proof that
the above can be constructed in time $T$.

Now we show that any entry of any map in the
tester can be constructed and computed in time~$\tilde O(t+p^{1/2})$.
If $t\le w$ then the result follows from Theorem~\ref{poly2}. Now suppose $t\ge w$.
Notice that $d/\epsilon\le q=\tilde O(1)$ with respect to $t$ and therefore
$O(s+p^{1/2})=\tilde O(t)$. This completes the proof.
\end{proof}

We now prove
\begin{theorem}
Let $q\ge d+1$, $34\ge c\ge 1+34/q$ and $\epsilon>0$ such that $$34\frac{d}{q}\ge \epsilon=c\frac{d}{q}\ge  \frac{d}{q}+\frac{34d}{q^2}.$$ A
$(\P(\FF_q,n,d),\FF_{q^t},\FF_q)$-$\bepsilon$-tester of size
$$s=\frac{1}{(1-c)^4}\left(\frac{d}{\epsilon}\right)^5\cdot t$$ can be
constructed in deterministic polynomial time $\tilde O(s)$ and any entry of any map in the
tester can be constructed and computed in time~$\tilde O(t+p^{1/2})$..
\end{theorem}
\begin{proof} Let $\epsilon_1=d/q$ and
$\epsilon_2=(c-1)d/q$. By Lemma~\ref{Triv}, (\ref{ctb11pl2}),
Corollary~\ref{Cqdt} and Theorem~\ref{LarE} we have
\begin{eqnarray*}
\nu^\cP_{\FF_{q}} (d,\FF_{q^{t}},\bepsilon)
&\le&
\nu^\cP_{\FF_{q}}(d,\FF_{q^{2t}},\overline{\epsilon_1+\epsilon_2})\\
&\le&\nu^\cP_{\FF_{q}}(d,\FF_{q^{2}},\overline{\epsilon_1})\cdot
\nu^\cP_{\FF_{q^2}}(d,\FF_{q^{2t}},\overline{\epsilon_2})\\
&\le&q\cdot \left(\frac{d}{\epsilon_2}\right)^4t\\
&\le&q\cdot \left(\frac{d}{(c-1)(d/q)}\right)^4t\\
&\le&\frac{q^5}{(c-1)^4}t=O\left(\frac{1}{(c-1)^4}\left(\frac{d}{\epsilon}\right)^5\right)\cdot
t.
\end{eqnarray*}
Notice here that $s\ge d/\epsilon \ge q/34\ge p/34$. This is why $p^{1/2}$ does
not appear in the complexity.

The complexity of constructing and computing any entry of any map
in the tester follows from Corollary~\ref{Cqdt}
and Theorem~\ref{LarE}.
\end{proof}

\subsection{Dense Testers for any $t$ and Small $q$}\labell{s21}

The following is Theorem \ref{T1} with the time complexity of constructing such
tester

\begin{theorem} \labell{T1p} Let $q<d^{1/2}$ be a power of prime
and $t$ be any integer.
Let $r$ be an integer such that $q^{2^{r-1}}<9d\le q^{2^r}$. Let
$\bfepsilon=(\epsilon_0,\ldots,\epsilon_{r-1},\epsilon_r)$ where
$\epsilon_i(q^{2^i}+1)\le q^{2^i}$ is an integer for
$i=0,1,\ldots,r-1$ and $2/3\ge \epsilon_r\ge 1/3$. Let
$$c_{q,\bfepsilon}:=
\sum_{i=0}^{r-1} \frac{\log(q^{2^i}+1)}{\epsilon_i(q^{2^i}+1)}$$
and
$$\pi_{q,\bfepsilon}:=\sum_{i=0}^{r-1}
\frac{-\log(1-\epsilon_i)}{\epsilon_i(q^{2^i}+1)}.$$

Then for
$$\overline{\epsilon^\star}:=\bepsilon_r\prod_{i=0}^{r-1}\bepsilon_i^{\lceil
d/(\epsilon_i(q^{2^i}+1))\rceil}\ge
\frac{2^{-\pi_{q,\bfepsilon}\cdot d}}{\Theta(d^2)},$$
a $(\HLF(\FF_q,n,d),\FF_{q^t},\FF_q)$-$\epsilon^\star$-tester of size
$$s
:= \Theta(d^6)\cdot 2^{c_{q,\bfepsilon}\cdot d}\cdot t$$
can be constructed in time $\tilde O(s)=\tilde O(2^{c_{q,\bfepsilon}\cdot d}\cdot t)$.
The time complexity of constructing and computing any entry in any map
in the tester is equal to $\tilde O(c_{q,\bfepsilon}d+t)$.
\end{theorem}
\begin{proof} We will go over the construction and compute the total time and the time for
constructing and computing any entry of any map. We have used the following results that
follows from Lemma \ref{Triv},
Corollary \ref{ctb11}, Lemma~\ref{basn} and~\ref{qdt}.
\begin{enumerate}
\item \label{P0k} $\nu_{\FF_q}(d,\FF_{q^{t_1}},\bepsilon)\le
\nu_{\FF_q}(d,{\FF_{q^{t_1t_2}}},\bepsilon)$.

\item\label{P1k}
$\nu_{\FF_q}(d,\FF_{q^{t_1t_2}},\bepsilon_1\bepsilon_2)\le
\nu_{\FF_q}(d,{\FF_{q^{t_1}}},\bepsilon_1)\cdot
\nu_{\FF_{q^{t_1}}}(d,\FF_{q^{t_1t_2}},\bepsilon_2).$

\item\label{P2k}
$\nu_{\FF_q}(d_1+d_2,\FF_{q^t},\bepsilon_1\bepsilon_2)\le
\nu_{\FF_q}(d_1,\FF_{q^t},\bepsilon_1)\cdot
\nu_{\FF_q}(d_2,\FF_{q^t},\bepsilon_2).$

\item\label{P3k} $\nu_{\FF_q}(d,\FF_{q^2},\bepsilon)\le q+1$ for any $\epsilon<1$ such that
$\epsilon(q+1)$ is an integer and $d\le \epsilon(q+1)$.
\end{enumerate}

Let $\eta_i=\epsilon_i(q^{2^i}+1)$. Then the following (from the proof of Theorem~\ref{T1})
shows how to construct such tester
\begin{eqnarray}
\nu_{\FF_q}(d,\FF_{q^t},\overline{\epsilon^\star})&\le &
\nu_{\FF_q}(d,\FF_{q^{t\cdot 2^r}},\overline{\epsilon^\star})\ \ \
\mbox{ By (\ref{P0}.)}\label{step1}\\
&\le & \left(\prod_{i=0}^{r-1}\nu_{\FF_{q^{2^i}}}
\left(d,\FF_{q^{2^{i+1}}},\bepsilon_i^{\lceil
d/\eta_i\rceil}\right) \right)
\nu_{\FF_{q^{2^r}}}(d,\FF_{(q^{2^r})^t},\bepsilon_r)\ \mbox{\ \ By (\ref{P1}.)}\label{step2} \\
&\le &\left(\prod_{i=0}^{r-1}
\nu_{\FF_{q^{2^i}}}(\eta_i,\FF_{q^{2^{i+1}}},\bepsilon_i)^{\lfloor
d/\eta_i\rfloor}
\cdot \nu_{\FF_{q^{2^i}}}(d-\eta_i{\lfloor
d/\eta_i\rfloor},\FF_{q^{2^{i+1}}},\bepsilon_i)\right)\nonumber \\
& & \ \ \ \ \ \ \ \ \cdot \nu_{\FF_{q^{2^r}}}(d,\FF_{(q^{2^r})^t},
\bepsilon_r)\
\mbox{\ \ By (\ref{P2}.)}\label{step3}\\
&\le& \left(\prod_{i=0}^{r-1}\left(q^{2^i}+1\right)^{\lceil
d/\eta_i\rceil}\right)
\nu_{\FF_{q^{2^r}}}(d,\FF_{(q^{2^r})^t},\bepsilon_r)
\mbox{\ \ By (\ref{P3}.)}\label{step4}\\
&\le&\left(\left(\prod_{i=0}^{r-1}\left(q^{2^i}+1\right)\right)
2^{c_{q,\bfepsilon}\cdot
d}\right)\nu_{\FF_{q^{2^r}}}(d,\FF_{(q^{2^r})^t},\bepsilon_r)\nonumber \\
&\le& \frac{q^{2^r}-1}{q-1}
\nu_{\FF_{q^{2^r}}}(d,\FF_{(q^{2^r})^t},\bepsilon_r)\cdot 2^{c_{q,\bfepsilon}\cdot d}\nonumber \\
&\le& \Theta(d^6)\cdot 2^{c_{q,\bfepsilon}\cdot d }\cdot t\mbox{\
\ By Theorem~\ref{LarE}}\label{step5}
\end{eqnarray}

In (\ref{step1}) and (\ref{step2})
we need to construct $\FF_{q^{2^r t}},\FF_{q^{2^{r-1}t}},\ldots,\FF_{q^{2t}}$ from $\FF_{q^{t}}$
which by Lemma~\ref{FF1} takes time $\tilde O(r(p^{1/2}2^{3r}+2^{4r}))$. Since
$p<q<d^{1/2}$, $2^r\le 2\log_q(9d)$ and $c_{q\bfepsilon}\ge 1/q$ (see Corollary~\ref{Co1})
the time complexity of (\ref{step1}) is $\tilde O(p^{1/2})=\tilde O(c_{q,\bfepsilon}d)$. In (\ref{step2}), by Lemma~\ref{ctb1},
the time of the construction is linear in the sum of time of the construction of each
tester $\FF_{q^{2^{i+1}}}\to \FF_{q^{2^{i}}}$
and in the size which is the product of the sizes. Constructing and computing
any entry in any map is linear in the sum of constructing and computing any entry of any map in each tester.
The same is true for~(\ref{step3}). In (\ref{step4}), by Lemma~\ref{qdt},
the testers that map $\FF_{q^{2^{i+1}}}$ to
$\FF_{q^{2^i}}$, $i=1,\ldots,r-1$, are constructed in time $\tilde O(\eta_i/\epsilon_i)=\tilde O(d^2)$
and constructing and computing any entry of any map in the testers takes time $\tilde O(1)$.
Computing each map in this tester involves
substituting an element of $\FF_{q^{2^{i}}}$ in a quadratic polynomial which takes time $poly(2^i,\log q)=
\tilde O(1)$. In (\ref{step5}) we use Theorem~\ref{LarE} (rather than Lemma~\ref{LT03})
that takes construction time $\tilde O(d^4t+p^{1/2})=\tilde O(d^4 t)$.
Constructing and computing any entry of any map in this tester takes time $\tilde O(t+p^{1/2})=\tilde O(t+c_{q,\bfepsilon}d)$. Now the time for the construction is clearly equal to $O(poly(d)\times s)$ where $s$ is the size
of the tester and therefore is equal to $\tilde O(2^{c_{q,\bfepsilon}\cdot d }\cdot t)$.

Using $T(\cdot)$ for the time of
constructing and computing any entry in any map in the tester, by the above discussion, we have
\begin{eqnarray*}
T_{\FF_q}(d,\FF_{q^t},\overline{\epsilon^\star})&= &
T_{\FF_q}(d,\FF_{q^{t\cdot 2^r}},\overline{\epsilon^\star})+\tilde O(c_{q,\bfepsilon}d)\\
&= & \left(\sum_{i=0}^{r-1}T_{\FF_{q^{2^i}}}
\left(d,\FF_{q^{2^{i+1}}},\bepsilon_i^{\lceil
d/\eta_i\rceil}\right) \right)+
T_{\FF_{q^{2^r}}}(d,\FF_{(q^{2^r})^t},\bepsilon_r)+O(r)+\tilde O(c_{q,\bfepsilon}d) \\
&= &\left(\sum_{i=0}^{r-1}
{\lfloor
d/\eta_i\rfloor}\cdot T_{\FF_{q^{2^i}}}(\eta_i,\FF_{q^{2^{i+1}}},\bepsilon_i)
+ T_{\FF_{q^{2^i}}}(d-\eta_i{\lfloor
d/\eta_i\rfloor},\FF_{q^{2^{i+1}}},\bepsilon_i)\right) \\
& & \ \ \ \ \ \ \ \ + T_{\FF_{q^{2^r}}}(d,\FF_{(q^{2^r})^t},
\bepsilon_r)+\tilde O(c_{q,\bfepsilon}d)\\
&=&\tilde O(c_{q,\bfepsilon}d)+ T_{\FF_{q^{2^r}}}(d,\FF_{(q^{2^r})^t},
\bepsilon_r)+\tilde O(c_{q,\bfepsilon}d)\\
&=&\tilde O(c_{q,\bfepsilon}d)+ \tilde O(t+c_{q,\bfepsilon}d)+\tilde O(c_{q,\bfepsilon}d)=\tilde O(c_{q,\bfepsilon}d+t).\\
\end{eqnarray*}
\end{proof}

\newpage
\section{Appendices}
\subsection{Appendix C}
We remind the reader that
$\bflambda=(\lambda_1,\lambda_2,\ldots,\lambda_t)\in \FF_q^t$ is of period $t$ if
$$\bflambda^0:=\bflambda,\ \bflambda^1:=(\lambda_t,\lambda_1,\ldots,\lambda_{t-1}), \ \bflambda^2:=(\lambda_{t-1},\lambda_t,\lambda_1,\ldots,\lambda_{t-2}),\cdots,
\bflambda^{t-1}:=(\lambda_2,\lambda_3,\ldots,\lambda_{t},\lambda_1)$$ are distinct.

By the proof of Lemma~\ref{ManI} it is enough to find a total order on $r=q^{t-2}/2t$
vectors $\bflambda\in \FF_q^t$ of period $t$ and show how to access the $m$th vector in time $\tilde O(\log m+t^2)$.

Define ${\cal S}$ the set of vectors $(0,0,\stackrel{k}{\ldots},0,\alpha_1,\ldots,\alpha_{t-k})$
where no $k$ consecutive zeros occurs in $(\alpha_1,\ldots,\alpha_{t-k})$. The integer $k$ will be determined later.
The following result is trivial
\begin{claim} The vectors in ${\cal S}$ are of period $t$.
\end{claim}

Let $M(k,n)$ be the number vectors $\bfalpha=(\alpha_1,\ldots,\alpha_{n})$ where no $k$ consecutive
zeros occurs in $\bfalpha$. We denote the set of all such vectors by $S(k,n)$. Notice that
${\cal S}=\{0\}^k\times S(k,t-k)$. Then
\begin{claim} We have: $M(k,n)=q^n$ for $n\le k-1$, $M(k,k)=q^k-1$ and
\begin{eqnarray}\label{hhh}
M(k,n)=q\cdot M(k,n-1)-(q-1)\cdot M(k,n-k-1).
\end{eqnarray}
Also
$$M(k,n)=(q-1)\cdot \sum_{i=1}^k M(k,n-i).$$
\end{claim}
\begin{proof} The number of vectors in $S(k,n-1)$ that ends with one of the vectors
in $(\FF_q\backslash \{0\})\times \{0\}^{k-1}$ is $(q-1)\cdot M(k,n-k-1)$. Denote the set of such vectors by $S'(k,n-1)$.
Notice that $S(k,n)=S'(k,n-1)\times (\FF_q\backslash\{0\})\cup (S(k,n-1)\backslash S'(k,n-1))\times \FF_q$.
This implies the first result.

For the second result notice that
\begin{eqnarray}\label{ordq}
S(k,n)=\bigcup_{i=1}^{k} \{0\}^{i-1}\times (\FF_q\backslash \{0\})\times S(k,n-i).
\end{eqnarray}
\end{proof}
We now give some lower bound for $M(k,n)$.
\begin{claim} We have
$$M(k,n)\ge q^n-n\cdot q^{n-k}.$$
\end{claim}
\begin{proof} Follows from (\ref{hhh}) and $M(k,n-k-1)\le q^{n-k-1}$ by induction.
\end{proof}
In particular,
\begin{claim} For $k=\lceil \log_qt\rceil+1$ we have
$$|{\cal S}|\ge \frac{q^{t-2}}{2t}.$$
\end{claim}
\begin{proof} We have
\begin{eqnarray*}
|{\cal S}|&=&M(k,t-k)\ge q^{t-k}-(t-k)q^{t-2k}\\
&=& q^{t-k}\left(1- \frac{t-k}{q^k}\right)\ge \frac{q^t}{2q^k}\ge \frac{q^{t-2}}{2t}.
\end{eqnarray*}
\end{proof}

Define any total order on $\FF_q$
where accessing the $i$th element
takes time $\log q$. Let $\alpha_1,\ldots, \alpha_{q-1}$ be the non-zero elements
of $\FF_q$ in that order. The following procedure defines a total order on $S(k,n)$
and therefore on ${\cal S}$ when $n=t-k$.
We denote the procedure that returns the $r$th element
in $S(k,n)$ by ${\bf Select}(n,r)$.
We define the order recursively using (\ref{ordq}). That is,
we first compute $M(k,i)$ for all $i=1,\ldots,n$ using~(\ref{hhh}).
Find $j_1\in \{1,2,\ldots,k\}$ such that
$$(q-1)\cdot \sum_{i=1}^{j_1-1} M(k,n-i)< r\le (q-1)\cdot \sum_{i=1}^{j_1} M(k,n-i)$$
Then for
$$r':=r- (q-1)\cdot \sum_{i=1}^{j_1-1} M(k,n-i)$$
find $j_2\in \{1,2,\ldots,q-1\}$ such that
$$(j_2-1)\cdot M(k,n-j_1)< r'\le j_2 \cdot M(k,n-j_1).$$
Then for
$$r'':=r'-(j_2-1)\cdot M(k,n-j_1)$$ define the element
$$\{0\}^{j_1-1}\times\{\alpha_{j_2}\}\times {\bf Select}\left (n-j_1,r''\right).$$
Since $M(k,i)=q^{\Theta(i)}$, computing $M(k,i)$, $i=1,\ldots,n$, takes time $\tilde O(n^2)$.
Computing $(q-1)\sum_{i=1}^j M(k,n-i)$ at each stage to find $j_1$
takes time $\tilde O(kn)=\tilde O(n)$.
To find $j_2$ at each stage we perform binary search for $j_2$. This takes time $\tilde O(n)$.
Therefore, the total time complexity is
$\tilde O(n^2)=\tilde O(t^2)$.

\subsection{Appendix B}

In this Appendix we prove
\begin{proposition}\label{Prop1}
Consider $$\epsilon(m)=\left(1-\frac{m}{q+1}\right)^{1/m},$$ and
$$\nu(m)=(q+1)^{1/m}.$$  Then
\begin{enumerate}
\item For $m=1$ we have
$$\epsilon(m)=1-\frac{1}{q+1} \ \ \mbox{and}\ \ \nu(m)=(q+1)=2^{\log(q+1)}.$$

\item \label{clogq} For $m=(1/c)\log(q+1)$ where $c=o(\log(q+1))$
we have
$$\epsilon(m)=1-\frac{1}{q+1}-\Theta\left( \frac{\log(q+1)}{c(q+1)^2}\right)
 \ \ \mbox{and}\ \  \nu(m)=2^{{c}}.$$

\item \label{loga} For $m=o(q)$ and $m=\omega(\log(q+1))$ we have
$$\epsilon(m)=1-\frac{1}{q+1}-\Theta\left( \frac{m}{(q+1)^2}\right)$$
and $$\nu(m)=(q+1)^\frac{1}{m}=
1+\frac{\ln(q+1)}{m}+\Theta\left(\frac{\log^2(q+1)}{m^2}\right).$$

\item  \label{bbb} For $m=(q+1)/2$ we have
$$\epsilon(m)=1-\frac{2\ln 2}{q+1}+\Theta\left( \frac{1}{(q+1)^2}\right)$$ and $$\nu(m)=(q+1)^\frac{2}{q+1}=
1+\frac{2\ln(q+1)}{q+1}+\Theta\left(\frac{\log^2(q+1)}{(q+1)^2}\right).$$

\item  For $m=c(q+1)$, $c$ constant we have
$$\epsilon(m)=1-\frac{\ln (1/(1-c))}{c(q+1)}+\Theta\left( \frac{1}{(q+1)^2}\right)$$ and $$\nu(m)=(q+1)^\frac{1}{c(q+1)}=
1+\frac{\ln(q+1)}{c(q+1)}+\Theta\left(\frac{\log^2(q+1)}{(q+1)^2}\right).$$

\item For $m=(q+1)-(q+1)/c$ where $c=\omega(1)$ we have
$$\epsilon(m)=1-\frac{\ln c}{q+1}+
\Theta\left(\frac{\ln^2 c}{2(q+1)^2}-\frac{\ln
c}{c(q+1)}\right)+\Theta\left(\frac{\ln^2 c}{c(q+1)^2}\right)$$
and
$$\nu(m)=(q+1)^\frac{1}{(q+1)-(q+1)/c}=
1+\frac{\ln(q+1)}{(q+1)}+
\Theta\left(\frac{\ln(q+1)}{c(q+1)}+\frac{\ln^2(q+1)}{(q+1)^2}\right).$$

\item \label{eee} For $m=(q+1)-c$, where $c=(q+1)^{o(1)}$, we have
$$\epsilon(m)=1-\frac{\ln (q+1)}{q+1}+
\Theta\left( \frac{\log c}{q+1}\right)$$ and
$$\nu(m)=
1+\frac{\ln(q+1)}{q+1}+\Theta\left(\frac{c\log(q+1)}{(q+1)^2}+\frac{\log^2(q+1)}{(q+1)^2}\right).$$
\end{enumerate}
\end{proposition}
\begin{proof}{\it Sketch.} For {\it \ref{clogq}.} we use
\begin{eqnarray}\label{exp} (1+x)^\alpha =
\sum_{n=0}^\infty {\alpha \choose n} x^n=1+\alpha x-\Theta(\alpha
x^2)\quad\mbox{ for } |x| < 1, \alpha<1/2
\end{eqnarray}
where $${\alpha\choose n} = \prod_{k=1}^n \frac{\alpha-k+1}k =
\frac{\alpha(\alpha-1)\cdots(\alpha-n+1)}{n!}.$$

For {\it \ref{loga}.} we use (\ref{exp}) and
\begin{eqnarray} \label{ex} {e}^{x} = \sum^{\infty}_{n=0} \frac{x^n}{n!} = 1 +
x + \frac{x^2}{2!} + \frac{x^3}{3!} + \cdots=1+x+\Theta(x^2)\text{
for } |x|<1.\end{eqnarray}

For {\it \ref{bbb}-\ref{eee}} we use (\ref{ex}) and
$$\frac{1}{1-x} = \sum^{\infty}_{n=0} x^n=1+x+\Theta(x^2)\text{ for }|x| < 1.$$
\end{proof}

The following table ignores the small terms
\begin{center}
\begin{tabular}{|l|l|c|}
$m$ & $\epsilon(m)$ & $\nu(m)$ \\
\hline\hline $m=1$ & $1-\frac{1}{q+1}$ & $q+1$\\
\hline $m=\log(q+1)/c$, $c=o(\log(q+1))$ & $1-\frac{1}{q+1}$ & $2^c$\\
\hline $m=o(q)$,\ $\omega(\log(q+1))$ & $1-\frac{1}{q+1}$ &
$1+\frac{\ln(q+1)}{m}$\\
\hline $m=c(q+1)$, $c<1$, $c=\Theta(1)$ &
$1-\frac{\ln(1/(1-c))}{c(q+1)}$ &
$1+\frac{\ln(q+1)}{c(q+1)}$\\
\hline $m=(q+1)-(q+1)/c$, $c=\omega(1)$ & $1-\frac{\ln c}{q+1}$ &
$1+\frac{\ln(q+1)}{q+1}$\\
\hline $m=(q+1)-c$, $c=\Theta(1)$ & $1-\frac{\ln(q+1)}{q+1}$ &
$1+\frac{\ln(q+1)}{q+1}$\\
 \hline
\end{tabular}
\end{center}

\subsection{Appendix C}\labell{s22}
In Theorem~\ref{T1} we showed the following. Let $q<d+1$ and $t$
be any integer. Let $r$ be an integer such that $q^{2^{r-1}}<9d\le
q^{2^r}$. Let
$\bfepsilon=(\epsilon_0,\ldots,\epsilon_{r-1},\epsilon_r)$ where
$\epsilon_i(q^{2^i}+1)\le q^{2^i}$ is an integer for
$i=0,1,\ldots,r-1$ and $2/3\ge \epsilon_r\ge 1/3$. Let
$$c_{q,\bfepsilon}:=
\sum_{i=0}^{r-1} c_{q,\bfepsilon,i} \mbox{\ \ \ where \ \
 \ } c_{q,\bfepsilon,i}:=\frac{\log(q^{2^i}+1)}{\epsilon_i(q^{2^i}+1)}$$ and
$$\pi_{q,\bfepsilon}:=
\sum_{i=0}^{r-1} \pi_{q,\bfepsilon,i} \mbox{\ \ \ where \ \
 \ } \pi_{q,\bfepsilon,i}:=\frac{-\log(1-\epsilon_i)}{\epsilon_i(q^{2^i}+1)}.$$

Then for
$$\overline{\epsilon^\star}:=\bepsilon_r\prod_{i=0}^{r-1}\bepsilon_i^{\lceil
d/(\epsilon_i(q^{2^i}+1))\rceil}\ge
\frac{2^{-\pi_{q,\bfepsilon}\cdot d}}{\theta(d^2)},$$ we have
$$\nu_{\FF_q}(d,\FF_{q^t},\overline{\epsilon^\star})
\le \theta(d^5)\cdot 2^{c_{q,\bfepsilon}\cdot d}\cdot t.$$

Now our goal in this appendix is to fix $\pi_{q,\bfepsilon}$ and
minimize $c_{q,\bfepsilon}$ or to fix $c_{q,\bfepsilon}$ and
minimize $\pi_{q,\bfepsilon}$. Therefore we define
$$c_q(\pi)=\min_{\pi_{q,\bfepsilon}=\pi} c_{q,\bfepsilon} \mbox{\ \ and\ \ }
\pi_q(c)=\min_{c_{q,\bfepsilon}=c} \pi_{q,\bfepsilon}.$$ To find
$c_q(\pi)$ we use the method of Lagrange multipliers. Consider
$$F_q(\bfepsilon,\lambda)=\pi_{q,\bfepsilon}-\lambda (c_{q,\bfepsilon}-c).$$
We have
\begin{eqnarray}\labell{Lln}
\frac{\partial F_q}{\partial \epsilon_i}=0 \Longrightarrow \lambda
=-\frac{L(\epsilon_i)}{\ln (q^{2^i}+1)}
\end{eqnarray}
where $$L(\epsilon)=\frac{\epsilon}{1-\epsilon}+\ln
(1-\epsilon)=\sum_{j=2}^\infty \left(1-\frac{1}{j}\right)
\epsilon^j.$$ The function $L:[0,1]\to \R$ is monotonically
increasing function, $L(0)=0$ and $L(1)=+\infty$. Therefore the
inverse function $L^{-1}:\R^+\to [0,1]$ is well defined and
monotonically increasing function. By (\ref{Lln}) we have
$$\epsilon_i=L^{-1}\left(
\alpha_i L(\epsilon_0) \right) \mbox{\ \ where \ \ }
\alpha_i=\frac{\log (q^{2^i}+1)}{\log (q+1)}.$$

Since $L$ and $L^{-1}$ are monotonically increasing functions, we
have
$$\epsilon_i=L^{-1}\left(
\alpha_i L(\epsilon_0) \right)\ge L^{-1}\left( L(\epsilon_0)
\right)=\epsilon_0.$$ For $\alpha>1$ and $\epsilon\sqrt{\alpha}<1$
we have
\begin{eqnarray*}
L^{-1}(\alpha L(\epsilon))&=&L^{-1}\left(\alpha \sum_{j=2}^\infty
\left(1-\frac{1}{j}\right)\epsilon^i\right)\\
&\le& L^{-1}\left(\sum_{j=2}^\infty
\left(1-\frac{1}{j}\right)(\sqrt{\alpha}\epsilon)^i\right)\\
&=&L^{-1}(L(\sqrt{\alpha}\epsilon))= \sqrt{\alpha}\epsilon
\end{eqnarray*} and for $\epsilon\sqrt{\alpha}\ge 1$ we
have $$L^{-1}(\alpha L(\epsilon))\le 1\le \sqrt{\alpha}\epsilon.$$
Therefore for any $\epsilon_0$ we have
\begin{eqnarray}\labell{ezei}
1\ge \frac{\epsilon_0}{\epsilon_i}\ge \frac{1}{\sqrt{\alpha_i}}.
\end{eqnarray}

Then, by (\ref{ezei}),
$$ c_{q,\bfepsilon,i}=\frac{\log(q^{2^i}+1)}{\epsilon_i(q^{2^i}+1)}=
c_{q,\bfepsilon,0} \cdot \alpha_i
\frac{\epsilon_0}{\epsilon_i}\frac{ q+1}{q^{2^i}+1}\ge
c_{q,\bfepsilon,0} \cdot \sqrt{\alpha_i} \frac{ q+1}{q^{2^i}+1}$$
and
$$c_{q,\bfepsilon,i}=\frac{\log(q^{2^i}+1)}{\epsilon_i(q^{2^i}+1)}=
c_{q,\bfepsilon,0} \cdot \alpha_i
\frac{\epsilon_0}{\epsilon_i}\frac{ q+1}{q^{2^i}+1}\le
c_{q,\bfepsilon,0} \cdot {\alpha_i} \frac{ q+1}{q^{2^i}+1}.$$
Therefore,
\begin{eqnarray}\labell{E1}
c_{q,\epsilon}=c_{q,\bfepsilon,0}\left(1+\frac{1}{\Theta(q)}\right).
\end{eqnarray}

We now give another bound that will be used in the sequel. Let
$t\ge 1$ be a real number such that $\epsilon_0=1-1/t$. Since
$L(1-1/t)=t-\ln t-1$, for any $t$ and $\alpha>1$ we have
$$L^{-1}\left(\alpha L\left(1-\frac{1}{t}\right)\right)=L^{-1}(\alpha t-\alpha\ln t-\alpha)
\le L^{-1}(\alpha t -\ln(\alpha t)-1)\le 1-\frac{1}{\alpha t}.$$
 Therefore
 \begin{eqnarray}\labell{talp}
 1-\frac{1}{t}=\epsilon_0\le
\epsilon_i\le1-\frac{1}{\alpha_i t}.\end{eqnarray}

To bound $\pi_{q,\bfepsilon,i}$ we first consider the function
$$\sigma(\epsilon)=\frac{\epsilon}{-\ln(1-\epsilon)}$$ for
$0\le\epsilon\le 1$. This function is monotonically decreasing and
for $0\le \epsilon\le 0.5$, $1\ge \sigma(\epsilon)>.5$.

Now by (\ref{ezei}) and the properties of $\sigma$ we have
$$\pi_{q,\bfepsilon,i}=\frac{-\log(1-\epsilon_i)}{\epsilon_i(q^{2^i}+1)}=
\pi_{q,\bfepsilon,0} \frac{\sigma(\epsilon_0)}{\sigma(\epsilon_i)}
\frac{q+1}{q^{2^i}+1}\ge \pi_{q,\bfepsilon,0}
\frac{q+1}{q^{2^i}+1}.$$ For the upper bound, let
$\epsilon_0=1-1/t$. We have two cases: The first case is when
$$t\ge 1+\frac{1}{2\sqrt{\alpha_i}-1}.$$
Then, by (\ref{talp}),
$$\frac{\sigma(\epsilon_0)}{\sigma(\epsilon_i)}\le
\frac{-\ln(1-\epsilon_i)}{-\ln(1-\epsilon_0)}\le
\frac{\ln(\alpha_it)}{\ln
t}\le\frac{\ln\alpha_i}{\ln(1+1/(2\sqrt{\alpha_i}-1))}+1\le
4\sqrt{\alpha_i}\ln\alpha_i.$$ and therefore
$$\pi_{q,\bfepsilon,i}=
\pi_{q,\bfepsilon,0} \frac{\sigma(\epsilon_0)}{\sigma(\epsilon_i)}
\frac{q+1}{q^{2^i}+1}\le \pi_{q,\bfepsilon,0} \cdot
(4\sqrt{\alpha_i}\ln\alpha_i) \frac{q+1}{q^{2^i}+1}.$$ The second
case is when $$t< 1+\frac{1}{2\sqrt{\alpha_i}-1}.$$ Then
$\epsilon_0<1/(2\sqrt{\alpha_i})<1/2$ and by (\ref{ezei}),
$\epsilon_i\le \sqrt{\alpha_i}\epsilon_0\le 1/2$. Then by the
properties of $\sigma$ we get
$$\pi_{q,\bfepsilon,i}=
\pi_{q,\bfepsilon,0} \frac{\sigma(\epsilon_0)}{\sigma(\epsilon_i)}
\frac{q+1}{q^{2^i}+1}\le \pi_{q,\bfepsilon,0} \cdot 2
\frac{q+1}{q^{2^i}+1}.$$ Therefore
\begin{eqnarray}\labell{E2}
\pi_{q,\epsilon}=\pi_{q,\bfepsilon,0}\left(1+\frac{1}{\Theta(q)}\right).
\end{eqnarray}

Now by (\ref{E1}) and (\ref{E2}) we get
$$c_{q,\bfepsilon}=\frac{\log (q+1)}{\epsilon_0 (q+1)}\left(1+\frac{1}{\Theta(q)}\right)$$
and
$$\pi_{q,\bfepsilon}=\frac{-\log(1-\epsilon_0)}{\epsilon_0(q+1)}\left(1+\frac{1}{\Theta(q)}\right).$$
This shows that the optimal solution (for large $q$) is determined
by the first term of $c_{q,\bfepsilon}$ and $\pi_{q,\bfepsilon}$.
Therefore (ignoring small terms) we get
$$\epsilon_0=\frac{\log
(q+1)}{c(q+1)}$$ and
$$\pi_q(c)=\min_{c_{q,\bfepsilon}=c}
\pi_{q,\bfepsilon}=\frac{-c\log\left(1-\frac{\log
(q+1)}{c(q+1)}\right)}{\log(q+1)}=\left\{
\begin{array}{ll}
\frac{c'\log(1-1/c')}{q+1} & c=c'\frac{\log(q+1)}{q+1}\\
\frac{\log e}{q+1} & c=\omega\left(\frac{\log(q+1)}{q+1}\right)
\end{array} \right..$$

To get a better bounds for small $q$ one can use the following
estimates
\begin{eqnarray*} L^{-1}\left(\alpha
L\left(1-\frac{1}{t}\right)\right)&=&L^{-1}(\alpha t-\alpha\ln
t-\alpha) \\
&\ge& L^{-1}((\alpha t-\alpha\ln t-\alpha+1)-\ln(\alpha
t-\alpha\ln
t-\alpha+1)-1)\\
&=&1-\frac{1}{\alpha t-\alpha \ln t-\alpha+1}.
\end{eqnarray*}
Therefore
$$\epsilon_i\ge 1-\frac{1}{\alpha_i t-\alpha_i \ln
t-\alpha_i+1}.$$

\end{document}